\documentclass[a4paper,UKenglish,cleveref,autoref]{lipics-v2021}

\pdfoutput=1

\title{Generalized Decidability via Brouwer Trees}

\author{Tom {de Jong}}{School of Computer Science, University of Nottingham, United Kingdom \and \url{http://www.tdejong.com}}{tom.dejong@nottingham.ac.uk}{https://orcid.org/0000-0003-1585-3172}{Royal Society [URF{\textbackslash}R1{\textbackslash}191055, URF{\textbackslash}R{\textbackslash}241007]}

\author{Nicolai Kraus}{School of Computer Science, University of Nottingham, United Kingdom \and \url{https://nicolaikraus.github.io}}{nicolai.kraus@nottingham.ac.uk}{https://orcid.org/0000-0002-8729-4077}{Royal Society [URF{\textbackslash}R1{\textbackslash}191055, URF{\textbackslash}R{\textbackslash}241007], Engineering and Physical Sciences Research Council [EP/Z000602/1]}

\author{Aref Mohammadzadeh}{School of Computer Science, University of Nottingham, United Kingdom \and \url{https://www.arefmohammadzadeh.com/}}{aref.mohammadzadeh@nottingham.ac.uk}{https://orcid.org/0009-0001-6884-7975}{}

\author{Fredrik {Nordvall Forsberg}}{Computer \& Information Sciences, University of Strathclyde, United Kingdom \and \url{https://fredriknf.com/}}{fredrik.nordvall-forsberg@strath.ac.uk}{https://orcid.org/0000-0001-6157-9288}{Engineering and Physical Sciences Research Council [EP/Y000455/2]}

\authorrunning{T. de Jong and N. Kraus and A. Mohammadzadeh and F. Nordvall Forsberg} 

\Copyright{Tom de Jong and Nicolai Kraus and Aref Mohammadzadeh and Fredrik Nordvall Forsberg} 

\ccsdesc[500]{Theory of computation~Type theory}
\ccsdesc[500]{Theory of computation~Constructive mathematics}

\keywords{Decidability in constructive mathematics, homotopy type theory, ordinals, Brouwer trees, countable choice}

\newcommand{\repo}{https://bitbucket.org/nicolaikraus/constructive-ordinals-in-hott/}
\newcommand{\baseurl}{https://arxiv.org/src/2602.10844v2/anc/html/Paper.html}

\supplement{Formalization in Cubical Agda}
\supplementdetails[swhid={swh:1:snp:af7a3d40f000eeb74aa0f06222f84f76e2046841;origin=https://bitbucket.org/nicolaikraus/constructive-ordinals-in-hott.git}]{Repository}{\repo}
\supplementdetails{HTML rendering}{\baseurl}

\acknowledgements{We are grateful to Chuangjie Xu for early discussions on the
  topic of \(\alpha\)\nobreakdash-decidability.}

\nolinenumbers 

\EventEditors{Claudia Faggian and Joost-Pieter Katoen}
\EventNoEds{2}
\EventLongTitle{41st Annual Symposium on Logic in Computer Science (LICS 2026)}
\EventShortTitle{LICS 2026}
\EventAcronym{LICS}
\EventYear{2026}
\EventDate{July 20--23, 2026}
\EventLocation{Lisbon, Portugal}
\EventLogo{}
\SeriesVolume{380}
\ArticleNo{59}

\crefformat{enumi}{(#2#1#3)}
\crefformat{enumii}{(#2#1#3)}
\crefformat{enumiii}{(#2#1#3)}
\crefformat{enumiv}{(#2#1#3)}

\usepackage[dvipsnames]{xcolor}
\usepackage{mathtools}
\usepackage{tikz-cd}

\newcommand{\bzero}{\mathsf{zero}}
\DeclareMathOperator{\bsuc}{\mathsf{succ}}
\DeclareMathOperator{\blim}{\mathsf{limit}}

\newcommand{\jump}[1]{#1^\uparrow}
\newcommand{\unjump}[1]{#1^\downarrow}
\newcommand{\doublejump}[1]{#1^\Uparrow}

\newcommand{\bbisim}{\mathsf{bisim}}
\newcommand{\lzero}{\mathord\leq\mbox{-}\mathsf{zero}}
\newcommand{\ltrans}{\mathord\leq\mbox{-}\mathsf{trans}}
\newcommand{\lsuccmono}{\mathord\leq\mbox{-}\mathsf{succ}\mbox{-}\mathsf{mono}}
\newcommand{\lcocone}{\mathord\leq\mbox{-}\mathsf{cocone}}
\newcommand{\llimiting}{\mathord\leq\mbox{-}\mathsf{limiting}}
\newcommand{\leqtrunc}{\mathord\leq\mbox{-}\mathsf{trunc}}
\newcommand{\brwtrunc}{\mathsf{trunc}}

\newcommand{\toincr}[1]{\xrightarrow{#1}}

\newcommand{\datakw}{\textcolor[HTML]{CD6600}{\mathsf{data}}}
\newcommand{\wherekw}{\textcolor[HTML]{CD6600}{\mathsf{where}}}
\newcommand{\AC}[1]{\textcolor[HTML]{008B00}{#1}}
\newcommand{\AT}[1]{\textcolor[HTML]{0000CD}{#1}}

\newcommand{\LPO}{\mathsf{LPO}}

\newcommand{\MP}{\mathsf{MP}}

\newcommand{\Prop}{\mathsf{Prop}}

\newcommand{\Set}{\mathsf{Set}}
\newcommand{\UU}{\mathcal{U}}
\newcommand{\N}{\mathbb N}
\newcommand{\Bool}{\mathbf{2}}
\newcommand{\Unit}{\mathbf{1}}

\newcommand{\bff}{\mathsf{false}}
\newcommand{\btt}{\mathsf{true}}
\newcommand{\Brw}{\mathsf{Brw}}
\newcommand{\Brwzl}{\Brw^{zl}}

\newcommand{\limMin}{\mathsf{limMin}}
\newcommand{\limMax}{\mathsf{limMax}}
\newcommand{\limMiF}[2]{\mathsf{limMin} \, {#1} \; {#2}}
\newcommand{\limMaF}[2]{\mathsf{limMax} \, {#1} \; {#2}}

\newcommand{\Psistruc}{\Psi^+}
\newcommand{\isSemiDec}{\mathsf{isSemiDec}}
\newcommand{\SemiDec}{\mathsf{SemiDec}}
\newcommand{\Psitrunc}{\Psi^*}
\newcommand{\FChoice}{\mathsf{FC}}

\newcommand{\defeq}{\mathrel{\vcentcolon\mspace{-1mu}\equiv}}

\newcommand{\Sierp}{\mathcal S}
\newcommand{\Svee}{\bigvee}

\newcommand{\psiapprox}[2]{\overline{\Psi_{#1}}(#2)}

\newcommand{\gendec}[2]{#1\mbox-\mathsf{decidable} \, #2}

\DeclareMathOperator{\upa}{\Uparrow}
\DeclareMathOperator{\dwa}{\Downarrow}

\DeclarePairedDelimiter{\trunc}{\|}{\|}
\DeclarePairedDelimiter{\totrunc}{|}{|}

\newcommand{\limMinRel}[3]{\limMin[#1,#2]{\sim}#3}
\newcommand{\limsucc}{\textsf{lim-succ}}

\newenvironment{fequation}[1]
{
  \newtagform{formalized-eq}{(\flink{#1}\,}{)}%
  \usetagform{formalized-eq}%
  \begin{equation}%
}
{
  \end{equation}%
  \usetagform{default}\ignorespacesafterend%
}

\newcommand{\formalized}{{\color{NavyBlue!75!White}{\raisebox{-0.5pt}{\scalebox{0.8}{\faCog}}}}}
\newcommand{\flinkurl}[1]{\href{#1}{\formalized}}
\newcommand{\flink}[1]{\flinkurl{\baseurl\##1}}

\begin{document}

\maketitle

\begin{abstract}
  In the setting of constructive mathematics, we suggest and study a framework
  for decidability of properties, which allows for finer distinctions than just
  ``decidable, semidecidable, or undecidable''.  We work in homotopy type theory
  and use Brouwer tree ordinals to specify the level of decidability of a
  property. In this framework, we express the property that a proposition is
  $\alpha$\nobreakdash-decidable, for an ordinal $\alpha$, and show that
  it generalizes decidability and semidecidability. Further generalizing known
  results, we show that $\alpha$\nobreakdash-decidable propositions are closed
  under binary conjunction, and discuss for which $\alpha$ they are closed under
  binary disjunction. We prove that if each $P(i)$ is semidecidable, then the
  countable meet $\forall i\in \N. P(i)$ is $\omega^2$-decidable, and similar
  results for countable joins and iterated quantifiers. We also discuss the
  relationship with countable choice.
  All our results are formalized in Cubical Agda.
\end{abstract}

\section{Introduction}

Decidability and semidecidability are at the heart of computer science.
For example, it is decidable whether a natural number $p$ is prime, and whether the pair $(p,p+2)$ is a twin prime pair.
Given a natural number $n$, it is semidecidable whether there is a twin prime pair above $n$, i.e., whether there is a number $p > n$ such that both $p$ and $p+2$ are prime.
A long-standing open problem in number theory, the famous \emph{twin prime conjecture}, asks whether there is a twin prime pair above $n$ for every natural number $n$.
Despite significant efforts and successes \cite{zhang2014bounded,polymath2014bounded}, the question remains unsolved.%
\footnote{For the avoidance of doubt, we do not make any progress on the twin prime conjecture --- it is merely 
  an example of a proposition in the context of (generalized) decidability. Moreover, once the conjecture 
  gets resolved, it becomes a decidable problem, so that we can no longer use it as an interesting example.} 

Furthermore, there is no known algorithm that could semidecide the twin prime conjecture.
While one could dovetail the searches and thus run the search for every $n$ at some point, a positive result would require all searches to be successful, which (at least on an intuitive level) would take ``too long.''
The current paper suggests and studies a framework in which the duration of such ``algorithms'' can be made precise and proved.

We work in constructive mathematics, where one commonly used~\cite{Rosolini1986,bauer2006first} definition of the statement ``the property $P$ is semidecidable'' is that there exists a binary sequence $s : \N \to \Bool$ such that $P$ holds if and only if there exists an index $i$ such that $s_i = 1$.
Checking whether such a binary sequence contains a $1$ corresponds to checking whether a given \emph{conatural number} (also known as \emph{extended natural number}) is finite; in fact, one way to define the set of conatural numbers is as the collection of non-decreasing binary sequences~\cite{Escardo2013}.
Classically, the ordinal corresponding to the set of conatural numbers is $\omega + 1$, and a semidecidable proposition is one for which a positive answer can be discovered in $\omega$-many steps, assuming this positive answer exists.

For the twin prime conjecture, there are $\omega$-many semidecidable sub-problems, namely one for each possible input $n$, and each of these sub-problems lets us discover a positive answer within $\omega$ steps (assuming this answer exists). Thus, one might expect (and we will prove) that the twin prime conjecture's ordinal is $\omega^2$.

Our general definition is that a proposition $P$ is $\alpha$-decidable, for a given ordinal~$\alpha$, if there exists an ordinal $x$ such that $P \leftrightarrow (x \geq \alpha)$.
Naturally, this definition and everything we can prove only matches our expectation if we choose a suitable notion of ordinals which ensures, for example, that the $(\omega+1)$-decidable propositions coincide with the semidecidable ones.
For instance, if we take our set of ordinals to be given by some set of ``syntactic'' representations~\cite{buchholz:notation,schuette:book,Takeuti1987}
such as Cantor Normal Forms (i.e., binary trees with a certain condition), then the inequality $x \geq \alpha$ is always decidable, forcing $\alpha$-decidability to be identical to decidability.
Similarly, if we take ordinals to mean ``sets with a wellfounded, transitive, and extensional relation''~\cite{Powell1975,AczelRathjen2010,Grayson1978,Grayson1982,hott-book,TypeTopologyOrdinals} in a constructive meta-theory, then the inequality $x \geq \omega + 1$ cannot be semidecidable.

A notion of ordinals that works for us are Brouwer tree ordinals~\cite{church:1938,kleene:notation-systems} as previously formulated \cite{kraus2023ordinals} in homotopy type theory, a setting that we adopt.
The type $\Brw$ of such ordinals is, together with the relation ${\leq} : \Brw \to \Brw \to \Prop$, given by a quotient inductive-inductive construction~\cite{altenkirch2016qiit}.
The constructors for $\Brw$ are $\bzero$, $\bsuc$, and $\blim$, constructing zero, successors, and limits of sequences, together with a path constructor to ensure that bisimilar sequences give equal limits, and constructors for the relation $\leq$ to ensure that the expected equations hold.
An important design choice for the current work is that the constructor $\blim$ only accepts \emph{strictly} increasing sequences as its argument.
As a consequence, it is decidable whether a given ordinal is finite or infinite: $\bzero$ is finite, $\blim f$ is always infinite, and $\bsuc \beta$ is finite if and only if $\beta$ is.
Similarly, it is semidecidable whether $\blim f \geq \omega + 1$ because the inequality holds if and only if we can find an index $i$ such that $f_i$ is infinite, which is equivalent to finding the digit $1$ in a binary sequence.

\subsubsection*{Setting}
We work in homotopy type theory, using standard terminology and notation~\cite{hott-book}.
In particular, we write $\trunc A$ for the propositional truncation, and $\exists x.\, P(x)$ for $\trunc{\Sigma x:X.\, P(x)}$.
If $P$ and $Q$ are propositions (i.e., types with at most one element), then we write $P \land Q$ instead of $P \times Q$, and $P \lor Q$ instead of $\trunc{P + Q}$. We write $\Prop$ for the subtype of propositions.%
\footnote{Everything can be assumed to take place in a fixed (e.g.\ the lowest) universe, and $\Prop$ is a subtype of this fixed universe. The Agda formalization is universe-parametric. Propositional resizing is not assumed.}

All the (small) types we consider are sets, i.e., have no interesting higher equalities.
Nevertheless, concepts from homotopy type theory are required for our development as we build on the framework introduced in \cite{kraus2021connecting,kraus2023ordinals},
which includes the set-truncated higher (a.k.a.\ quotient) inductive-inductive type $\Brw$ of Brouwer tree ordinals and its relation $\leq$.
While no further higher inductive types are used (apart from in \cref{subsec:Sierpinski}),
univalence is heavily relied on to establish the desired properties of $\Brw$~\cite{kraus2023ordinals}.

In principle, it should be possible to carry out our work in a plain version of Martin-L\"of Type Theory or another constructive foundation, as long as function extensionality holds, propositional truncation is available, and the type $\Brw$ with its relation~$\leq$ are assumed together with their properties. (\cref{subsec:Sierpinski} would require a similar additional assumption.)

\subsubsection*{Summary of main results}

Defining an arbitrary proposition $P$ to be \emph{$\alpha$\nobreakdash-decidable} if $\exists x:\Brw. \, P \leftrightarrow x \geq \alpha$,
we first establish that $\alpha$-decidability generalizes existing notions of decidability, in the sense that $P$ is decidable if and only if it is $1$-decidable (if and only if it is $\omega$-decidable), and semidecidable if and only if it is $(\omega + 1)$-decidable. We show that the set of $\alpha$-decidable propositions is closed under binary conjunctions, and the smaller set of $(\omega \cdot n + k)$-decidable propositions is closed under binary disjunctions.
We further study families $P : \N \to \Prop$ of semidecidable propositions and construct their \emph{characteristic ordinal}, which allows us to determine that their countable meet is $\omega^2$-decidable, an instance of which is the twin prime conjecture.
The countable join of the family is only $(\omega \cdot 3)$-decidable, unless countable choice is assumed which, as is well-known, makes it semidecidable.
In a related vein, we study the notion of \emph{Sierpi\'nski-semidecidability}, an alternative to semi\-decidability that is closed under countable joins without the assumption of countable choice, and its relation to $\alpha$-decidability.
Using constructive taboos such as the Limited Principle of Omniscience or Markov's principle, we are able to separate classes of $\alpha$-decidability, and studying the effect of quantifier alternations allows us to establish that the search for a counter-example to the twin prime conjecture, assuming it exists, is successful within $\omega^2+\omega$ steps.

\subsubsection*{Agda formalization}
We have formalized and checked all our results in Cubical Agda~\cite{cubicalAgdapaper}, building on a previous formalization of constructive notions of ordinals~\cite{kraus2023ordinals}.
Using Cubical Agda was convenient since it supports quotient inductive-inductive types natively (i.e., without tricks or encodings), but any proof assistant used for homotopy type theory (e.g.\ Rocq) would have worked.

The source code of the formalization can be found in the repository
\url{\repo}, type-checks with Agda 2.8.0, and is archived on Software Heritage (cf.~\emph{Supplementary Material} on the title page).
A~browsable rendering can be found at
\url{\baseurl}. Throughout this paper, each definition, lemma, proposition, and theorem is annotated with the symbol~\formalized, which is a link to the corresponding Agda code in that HTML file.

\subsubsection*{Related work}

Our work lies at the boundary of \emph{synthetic computability theory}, a subject pioneered by Bauer~\cite{bauer2006first} and influenced by Rosolini~\cite{Rosolini1986}, Bridges and Richman~\cite{Bridges_Richman_1987}, and many others.
The approach of synthetic computability theory is to \emph{not} fix a model of computation (such as Turing machines), but to use the notion of computability that is built into the logic one uses.
For example, in synthetic computability theory every function $f : \N \to \N$ is (tautologically) computable,
and a standard definition of semidecidability says that $P$ is semidecidable if there exists a binary sequence that contains the digit $1$ if and only if $P$ holds.

The theory of \emph{infinite time Turing machines} initiated by Hamkins and
Lewis~\cite{ittm} investigates a notion of Turing machine which can run
for infinite ordinal time. While there are similarities in motivation,
the resulting development is quite different from ours. Infinite time
Turing machines are a specific model of computation, couched in a
classical metatheory, whereas our notion of decidability is synthetic,
and developed constructively. It may be possible to describe our development
as a synthetic version of infinite time Turing machines that follows the approach
of constructive mathematics. However, the foundational differences mean that a
precise connection could only be made by interpreting our theory into an effective model, which, while of interest, is beyond the scope of this paper.

Similarly, \emph{Limit Computable Mathematics}~\cite{Gold,Hayashi} is concerned with problems that can be solved in the limit of a possibly infinite sequence of ``guesses'', i.e., by infinitely long decision procedures. This gives rise to constructive interpretations of classical proofs~\cite{Berardi}, which our framework of ordinal decidability is not aiming to do.

Within constructive type theories, Brouwer tree ordinals have been investigated a number of times,
e.g.\ by Coquand, Hancock and Setzer~\cite{CoquandHancockSetzer1997}.
The most relevant work for us is the one by Kraus, Nordvall Forsberg, and Xu~\cite{kraus2023ordinals}, which suggests the version of Brouwer tree ordinals that are used in this paper.
Quotient inductive-inductive types were introduced in the homotopy type theory book~\cite{hott-book} and further studied by Altenkirch and collaborators~\cite{altenkirch2016qiit}.

A preliminary variant of the definition of \(\alpha\)-decidability, and its relation to semidecidability, was presented in a talk at TYPES~2022~\cite{kraus2022decidability}. The present paper develops and refines the  theory suggested in that talk.

\subsubsection*{Contents}

We review Brouwer tree ordinals in \cref{sec:Brouwer}, and connect our generalized decidability with existing decidability notions in \cref{sec:ord-dec}.
\cref{sec:reduction-to-limits} shows that it suffices to consider $\alpha$-decidability with $\alpha$ a limit ordinal, a result that is of limited interest in itself but helpful for the remaining technical development.
\cref{sec:BinConj,sec:bin-disj} discuss binary conjunctions $\land$ and disjunctions $\lor$ of $\alpha$-decidable propositions, while \cref{sec:semidec-and-quantifiers} establishes results about countable meets $\forall$ and joins $\exists$. \cref{sec:CC} presents connections with countable choice, while \cref{sec:Agda-technical-discussion} explains formalization-specific strategies, and \cref{sec:conclusions} concludes.

\section{Brouwer Tree Ordinals}\label{sec:Brouwer}

The notion of ordinals that this work is based on are the Brouwer tree ordinals introduced in
Kraus, Nordvall Forsberg, and Xu~\cite{kraus2021connecting,kraus2023ordinals}.
In this section, we recall their constructions in two stages.
We first explain how the type $\Brw$ with the relation $\leq$ is constructed in \cref{subsec:constr-of-Brw}, and recall the important properties in \cref{subsec:Brw-properties}.
The reader is free to skip the first subsection as only the properties, but not the implementation, are referred to later.

\subsection{Construction of Brouwer Tree Ordinals}\label{subsec:constr-of-Brw}

The type of Brouwer tree ordinals, denoted $\Brw$, is defined mutually with its order relation $\leq \, : \Brw \to \Brw \to \Prop$ as a quotient inductive-inductive type~\cite{altenkirch2016qiit}. This means the type and its relation are generated simultaneously by a list of point- and path-constructors.
In the rest of the paper, whenever we talk of \emph{ordinals}, we refer to elements of the type $\Brw$.

The point constructors for the type $\Brw$ are $\bzero$, $\bsuc$, and $\blim$; the important thing to note here is that $\blim$ constructs the limit of a \emph{strictly increasing} sequence, meaning that the input of $\blim$ is a function $f : \N \to \Brw$ such that $\forall n. \, \bsuc(f_n) \leq f_{n + 1}$.
This ``design choice'' has important consequences with respect to the decidability of $\Brw$, which we will discuss in \cref{subsec:Brw-properties}.
The path constructors of $\Brw$ ensure that $\Brw$ is a set, and that limits of bisimilar sequences are equal.
Here, a sequence $f$ is simulated by a sequence $g$ if $\forall i. \, \exists j. \, f_i \leq g_j$, and two sequences are bisimilar if each simulates the other.
The constructors of the relation $\leq$ ensure that $\bzero$ is minimal, $\bsuc$ is monotone, $\blim$ really constructs limits, and the relation is transitive and valued in propositions.
The constructors are given in pseudo-Agda notation in \cref{fig:Brw-def}.

\begin{figure}
\[
  \begin{array}{lll}
  \multicolumn{3}{l}{\datakw \; \AT{\Brw} : \AT{\Set} \; \wherekw}\\
    \quad \AC{\bzero} &:& \AT{\Brw} \\
  \quad \AC{\bsuc} &:& \AT{\Brw} \to \AT{\Brw} \\
  \quad \AC{\blim} &:& (\AT{\N} \toincr{<} \AT{\Brw}) \to \AT{\Brw} \\
  \quad \AC{\bbisim}  &:&  f\, \AT{\approx}\, g \to \AC{\blim} \, f = \AC{\blim} \, g\\
  \quad \AC{\brwtrunc} &:& (p \, q : x\, =\, y) \to p\, =\, q
  \end{array}
\quad
\begin{array}{lll}
  \multicolumn{3}{l}{\datakw \; \leq {} : \AT{\Brw} \to \AT{\Brw} \to \AT{\Prop} \; \wherekw} \\
  \quad \AC{\lzero} &:& \AC{\bzero} \leq x \\
  \quad \AC{\ltrans} &:& x \leq y \to y \leq z \to x \leq z\\
  \quad \AC{\lsuccmono} &:&x \leq y \to \AC{\bsuc} \, x \leq \AC{\bsuc} \, y \\
  \quad \AC{\lcocone} &:& (\exists k. x \leq f\,k) \to x \leq \AC{\blim} \, f \\
  \quad \AC{\llimiting} &:& (\forall k. f\,k \leq x) \to \AC{\blim} \, f \leq x \\
  \quad \AC{\leqtrunc} &:& (p \, q : x \leq y) \to p = q
\end{array}
\]
\caption{The construction of $(\Brw, \leq)$ as a quotient inductive-inductive type as presented in \cite{kraus2023ordinals}.
The variables $x, y, z : \Brw$ and $f, g : \N \toincr{<} \Brw$ are implicitly quantified over.
Note that $\N \toincr{<} \Brw$ refers to strictly increasing sequences, and $f \approx g$ stands for bisimilarity.}\label{fig:Brw-def}
\end{figure}

For this definition, one can define addition, multiplication, and exponentiation, with the usual properties.
Specific examples such as finite ordinals $n$, the least infinite ordinal $\omega$, and~$\epsilon_0$ can be defined straightforwardly.
The details are given in
~\cite{kraus2023ordinals}.

\subsection{Properties of Brouwer Tree Ordinals}\label{subsec:Brw-properties}

The definition of $\Brw$ and its relation~$\leq$ have several useful consequences.
When proving a \emph{property} of ordinals, i.e., constructing an element of a proposition-valued type family over $\Brw$, it suffices to consider the point constructors ($\bzero$, $\bsuc$, $\blim$).
With the following lemma, we recall several core observations that we need for the current paper from \cite{kraus2023ordinals}.

\begin{lemma}[\flink{Lemma-1}]\label{lem:recall-basic-Brw-properties}
We have the following basic properties of $\Brw$ and its relation $\leq$, with $\alpha < \beta$ being defined as $\bsuc\,\alpha \leq \beta$.
\begin{romanenumerate}
	\item \label{item:classification}
	\emph{Classification of ordinals:} Every ordinal is either zero, or a successor, or a limit.

	\item \label{item:finite-is-dec}
	\emph{Decidability of finiteness:}
	Related to the classification above, it is decidable whether a given ordinal $\alpha$ is finite (i.e., $\alpha < \omega$) or infinite (i.e., $\alpha \geq \omega$).
	If $\alpha$ is finite, then we can explicitly find the natural number $n$ such that the canonical embedding $\N \to \Brw$ sends $n$ to $\alpha$.
	However, being able to decide $\alpha \leq \beta$ and $\alpha = \beta$ in general is a constructive taboo.

	\item \emph{Properties of the relation:} The relation $\alpha < \beta$ is transitive, wellfounded, and extensional.

          \newcounter{properties-counter}
          \setcounter{properties-counter}{\value{enumi}}
\end{romanenumerate}
The following properties characterize the relation $\leq$\textup{:}
\begin{romanenumerate}
  \setcounter{enumi}{\value{properties-counter}}
	\item \label{lem:suc-mono-inv}
	$\bsuc\,\alpha \leq \bsuc\,\beta$ if and only if $\alpha \leq \beta$.
	\item \label{lem:x-under-limit-means-x-under-element}
	If $\alpha < \blim f$ then $\exists n. \, \alpha < f_n$.
	\item \label{lem:lim-under-suc}
	If $\blim \, f \leq \bsuc\,\alpha$ then $\blim f \leq \alpha$, and if $\alpha < \blim f$ then $\bsuc\,\alpha < \blim f$.
	\item \label{lem:lim-under-lim-simulation}
	$\blim f \leq \blim g$ implies that $f$ is simulated by $g$. \qed
\end{romanenumerate}
\end{lemma}
We omit the (known) proofs of these properties.
Key for the classification of ordinals is the choice of only allowing limits of strictly increasing sequences.
This also gives decidability of finiteness, because $\bzero$ is finite, $\blim f$ is always infinite, and $\bsuc \beta$ is finite if and only if $\beta$ is.
Note that, in the formulation of the classification property, a limit ordinal is an ordinal $\alpha$ such that there \emph{exists} ($\| \Sigma \ldots \|$) a sequence~$f$; asking for the sequence $f$ itself would be too strong for the previous statement because of the constructor $\bbisim$.

In the literature on ordinals, some authors consider zero a limit.
Our definition does not allow us to do this, but it often happens that we want to express that an ordinal is either zero or a limit.
We write $\Brwzl$ for the subtype of $\Brw$ consisting of zero and limits.
Justified by the decidability property \cref{lem:recall-basic-Brw-properties}\cref{item:classification}, we refer to such ordinals as \emph{non-successor ordinals}.

\subsection{Rounding to Limits}

As a final preparation for the study of ordinal decidability, we note that we can \emph{round} every Brouwer tree ordinal to a non-successor ordinal.
We can round up or down, and we write $\upa, \dwa : \Brw \to \Brwzl$ for these two functions, defined in the obvious way:
\begin{alignat*}{4}
	&\upa \bzero &&\defeq \bzero \qquad\qquad &&\dwa \bzero &&\defeq \bzero \\
	&\upa (\bsuc x) &&\defeq x + \omega \qquad\qquad &&\dwa (\bsuc x) &&\defeq \dwa x \\
	&\upa (\blim f) &&\defeq \blim f  \qquad\qquad &&\dwa (\blim f) &&\defeq \blim f
\end{alignat*}
These definitions extend straightforwardly to the path constructors of $\Brw$.
The functions $\upa$ and $\dwa$ calculate non-successors, which is automatic for most cases; for the successor cases, note that $\upa (\bsuc x)$ is the limit of the sequence $\lambda n. x + n$, and $\dwa (\bsuc x)$ is a non-successor by induction on $x$.

\begin{lemma}[\flink{Lemma-2}]\label{cor:y<=xThen-RDy<=RDx}
	The functions $\upa$ and $\dwa$ preserve $\leq$. \qed
\end{lemma}

Using the rounding down function $\dwa$, we can split every ordinal as the sum of an infinite (or zero) and a finite ordinal:

\begin{lemma}[\flink{Lemma-3}]\label{item:spit-ordinal}
Every $\alpha$ can be written as $\alpha =  \lambda + n$, where $\lambda$ is a non-successor, and $n$ is finite (i.e., given by a natural number). \qed
\end{lemma}

Limit ordinals cannot distinguish between $\alpha$ and $\upa \alpha$, in the following sense:
\begin{lemma}[\flink{Lemma-4}]\label{lemma:RU23}
	For a non-successor ordinal $\lambda : \Brwzl$ and any ordinal $\alpha : \Brw$, we have:
	\begin{romanenumerate}
		\item\label{item:up-1} $\alpha \leq \lambda \longleftrightarrow \upa \alpha \leq \lambda$
		\item\label{item:up-2} $\lambda \leq \alpha \longleftrightarrow \lambda \leq \dwa \alpha$
		\item\label{item:up-3} $\lambda < \alpha \longleftrightarrow \lambda < \upa \alpha$ \quad (i.e., $\lambda + 1 \leq \alpha \longleftrightarrow \lambda + 1 \leq \upa \alpha$).
	\end{romanenumerate}
\end{lemma}
\begin{proof}
	These properties are all proved by straightforward inductive arguments.
\end{proof}

\section{Ordinal Decidability}\label{sec:ord-dec}

Based on the notion of ordinals recalled above, we define ordinal decidability:

\begin{definition}[\flink{Definition-5} $\alpha$-Decidability]\label{def:alpha-decidability}
	Given $\alpha : \Brw$, a proposition~$P$ is said to be \emph{$\alpha$-decidable} if there exists $y$ such that $P$ and $\alpha \leq y$ are equivalent:
	\begin{equation*}
		\gendec \alpha P \; \defeq \; \exists y : \Brw. (P \leftrightarrow \alpha \leq y).
	\end{equation*}
\end{definition}

\subsection{Connection with Existing Notions}

Ordinal decidability generalizes the notions of decidability and semidecidability in the following sense:

\begin{proposition}[\flink{Proposition-6}]\label{prop:DecOmegaEquivNew}
	A proposition $P$ is decidable if and only if it is $1$-decidable;
	and it is semidecidable if and only if it is $(\omega + 1)$-decidable:
	\begin{align*}
		P \lor \neg P \; & \longleftrightarrow \; \gendec 1 P \\
		\left(\exists (s : \N \to \Bool).(P \leftrightarrow \exists i. s_i = 1)\right) \; &\longleftrightarrow \; \gendec {(\omega + 1)} P
	\end{align*}
\end{proposition}

Before giving the proof, we recall the ``jumping sequence construction'' (as well as a straightforward extension) that allows switching between binary sequences and sequences of ordinals, as introduced in \cite[Thm 17]{kraus2023ordinals}:

\begin{definition}[\flink{Definition-7}]\label{def:jumping}
	If $s : \N \to \Bool$ is a binary sequence, the strictly increasing sequence $\jump s : \N \to \Brw$ is defined by
        \[
		\jump s_0 \defeq 0 \quad\text{and}\quad
		\jump s_{n+1} \defeq
		\begin{cases}
			\omega & \text{if $n$ is minimal such that $s_n = 1$} \\
			\jump s_n +1 & \text{else.}
		\end{cases}
        \]
	Given a sequence $t : \N \to \Brw$, the sequence $\unjump t : \N \to \Bool$ is defined by
	\begin{equation*}
		\unjump t_n \defeq
		\begin{cases}
			0 & \text{if $t_n$ is finite (a decidable property by \cref{lem:recall-basic-Brw-properties}\cref{item:finite-is-dec})} \\
			1 & \text{if $t_n$ is infinite.}
		\end{cases}
	\end{equation*}
\end{definition}

\begin{proof}[Proof of \cref{prop:DecOmegaEquivNew}]
	We only sketch the proof of the second property.

	``$\rightarrow$'':
	Assume that $P$ is semidecidable, witnessed by a sequence {$s : \N \to \Bool$}.
	We set $y \defeq \blim {\jump s}$.
	We claim $P \leftrightarrow \omega + 1 \leq y$: if $P$ holds, then there is $i$ with $s_i = 1$, thus $\jump {s_{i+1}} \geq \omega$ and therefore $\blim {\jump s} > \omega$.
	On the other hand, if $\omega + 1 \leq \blim {\jump s}$, then there exists $i$ such that ${\jump s_i} \geq \omega$ (\cref{lem:recall-basic-Brw-properties}\cref{lem:x-under-limit-means-x-under-element}), and we can find $1$ among the values $\{s_0, \ldots, s_{i-1}\}$.

	``$\leftarrow$'':
	Assume that $P$ is $(\omega + 1)$-decidable with witness $y : \Brw$.
	We want to construct $s : \N \to \Bool$ such that $\omega + 1 \leq y$ if and only if  $\exists i. s_i = 1$.
	We construct the sequence $s$ by case distinction on $y$:
	\begin{itemize}
		\item If $y \equiv \bzero$, we set $s$ to be constantly $0$.
		\item If $y \equiv \bsuc\,y'$, we check whether $y'$ is finite (\cref{lem:recall-basic-Brw-properties}\cref{item:finite-is-dec}); if so, we set $s$ to be constantly~$0$, and otherwise, constantly~$1$.
		\item If $y \equiv \blim f$, we set $s \defeq \unjump f$ which has the desired property due to \cref{lem:recall-basic-Brw-properties}\cref{lem:x-under-limit-means-x-under-element}. \qedhere
	\end{itemize}
\end{proof}

\begin{remark}[\flink{Remark-8}]\label{Remark-intuition}
	We note two points:
	\begin{romanenumerate}
		\item For sufficiently small ordinals $\alpha$, we can intuitively understand the statement ``$P$ is $\alpha$\nobreakdash-decidable'' as ``if evidence for~$P$ exists, then it can be discovered in less than $\alpha$ steps.''
		Thus, the second part of \cref{prop:DecOmegaEquivNew} means that a proposition is semidecidable if and only if its evidence, assuming it exists, can be discovered in at most $\omega$ steps.

		As one would intuitively expect, a proposition is decidable if and only if its evidence (or lack thereof) can be discovered in finitely many steps.
		Since our abstract framework does not consider concrete algorithms, it also cannot distinguish between numbers of steps that only differ by a finite amount (cf.~\cref{theorem:RUDec}), so the intuitive statement about decidable properties is captured by the first part of \cref{prop:DecOmegaEquivNew}.
		For the further development, we do not use the terminology involving steps and refer to the technically simpler concept of $\alpha$-decidability instead.

		\item As we will see in \cref{sec:reduction-to-limits}, it is sufficient to consider $\alpha$-decidability for non-successor ordinals $\alpha$; if we phrase \cref{prop:DecOmegaEquivNew} with this in mind, and extend it with the trivial observation that $0$-decidability simply means that a proposition holds, we get:
		\begin{align*}
			\text{$P$ holds} \; & \longleftrightarrow \; \text{$P$ is $(\omega \cdot 0)$-decidable}, \\
			\text{$P$ is decidable} \; & \longleftrightarrow \; \text{$P$ is $(\omega \cdot 1)$-decidable}, \\
			\text{$P$ is semidecidable} \; & \longleftrightarrow \; \text{$P$ is $(\omega\cdot 2)$-decidable}.
		\end{align*}
	\end{romanenumerate}
\end{remark}

\subsection{First Steps Towards a Hierarchy of Decidability}\label{subsec:first-steps-towards-hierarchy}

The intuition implied in the introduction and in \cref{Remark-intuition} is that the notion of $\alpha$-decidability gets weaker the larger $\alpha$ is, generalizing the fact that semidecidability is weaker than decidability.
One way to make this intuition precise would be to conjecture that, if $\alpha \leq \beta$, then any $\alpha$-decidable proposition is $\beta$-decidable.
Unfortunately, we do not know whether this statement holds in this generality.
We can prove the following first approximations:

\begin{proposition} [\flink{Proposition-9}]
For any $k: \N$, every $(\omega \cdot k)$-decidable proposition is $(\omega \cdot (k+1))$-decidable.
\end{proposition}
\begin{proof}
    Suppose $x: \Brw$ is evidence that $P$ is $(\omega  \cdot k)$-decidable. We want to prove that $x' = x + \omega$ is evidence that $P$ is $(\omega \cdot k + \omega)$-decidable. To do this,  it suffices to show that $(\omega  \cdot k) \leq x$ if and only if $(\omega  \cdot k + \omega) \leq x + \omega$. The ``only if'' case holds since addition is monotone, and the ``if'' case holds by induction on $k$ and $x$.
\end{proof}
\begin{proposition} [\flink{Proposition-10}]
For any $k: \N$, every $\omega ^ k$-decidable proposition is $\omega ^ {k+1}$-decidable.
\end{proposition}
\begin{proof}
    We prove this by showing that for any $k: \N$ and $x : \Brw$, $\omega^k \leq x$ if and only if $\omega^{k+1} \leq x \cdot \omega$. The ``only if'' case holds because multiplication on $\Brw$ is monotone, and the ``if'' case holds since multiplication by $\omega$ is left cancellable: for each pair of ordinals $x , y : \Brw$, if $\omega \cdot x \leq \omega \cdot y$ then $x \leq y$.
\end{proof}
Obtaining more general upward closure results is ongoing work. We are optimistic that it is true if $\alpha$ and $\beta$ are in Cantor Normal Form below $\epsilon_0$:
\begin{conjecture}[\flink{Conjecture-11}]
	For $\alpha,\beta : \Brw$ in Cantor Normal Form below $\epsilon_0$ such that $\alpha \leq \beta$, every $\alpha$-decidable proposition is $\beta$-decidable.
\end{conjecture}

\section{Reduction to Limit Ordinals}\label{sec:reduction-to-limits}

As we do not use an explicit notion of algorithm but work with an abstract definition of $\alpha$-decidability,
and the complexity of a function is not an internal property of type theory, it is unsurprising that the framework cannot separate $1$-decidable from $n$-decidable propositions, for any positive natural number $n$.
We generalize this observation and show that we cannot distinguish between ${(\alpha+1)}$\nobreakdash-decidability and $(\alpha + \omega)$-decidability; in other words, we do not lose anything if we restrict ourselves to $\lambda$-decidability for limit ordinals $\lambda$.

\begin{lemma}[\flink{Lemma-12}]\label{lemma:DecSucA}
	Let $P$ be a proposition and $\alpha$ be an ordinal. Then, $P$ is $(\alpha + 1)$-decidable if and only if it is $(\alpha + 2)$-decidable.
\end{lemma}
\begin{proof}
	By definition, the goal is to prove:
	\begin{equation*}
		\Big( \exists u . (P \leftrightarrow (\alpha+1) \leq u) \Big) \longleftrightarrow
		\Big( \exists v . (P \leftrightarrow (\alpha+2) \leq v) \Big).
	\end{equation*}
	The implications in both directions can be shown with the same strategy.
	We only sketch the latter.
	As we are proving a proposition, we can assume that we are given an explicit witness $v$ that $P$ is $(\alpha +2)$-decidable, and it suffices to consider the cases that $v$ is zero, a~successor, or a limit.
	The first is trivial; in the second, if $v = v' + 1$, we choose $u$ to be $v'$ and have $\alpha + 2 \leq v \longleftrightarrow \alpha + 1 \leq v'$ by \cref{lem:recall-basic-Brw-properties}\cref{lem:suc-mono-inv}. If $v$ is a limit, we choose $u$ to be $v$ and apply \cref{lem:recall-basic-Brw-properties}\cref{lem:lim-under-suc}.
\end{proof}

By a straightforward induction, we have:

\begin{corollary}[\flink{Corollary-13}]\label{theorem:SucOfSucIsOrdDecEquiv}
	For every ordinal $\alpha$ and natural number $n$, a proposition $P$ is $(\alpha+1)$-decidable if and only if it is $(\alpha+1+n)$-decidable. \qed
\end{corollary}

\begin{remark}[\flink{Remark-14}]\label{remark: SuccChangeOrdDecInLim}
	Note that \cref{lemma:DecSucA} cannot be formulated with $\alpha$ instead of $\alpha+1$ because, if $\alpha$ is a limit ordinal (or zero), the statement will not hold in general. For example, we have seen that $(\omega+1)$-decidability is just semidecidability, while $\omega$-decidability is decidability.
\end{remark}

To formalize the intuition that the hierarchy of ordinal decidability is indexed by limit ordinals,
our next goal is to show that \cref{theorem:SucOfSucIsOrdDecEquiv} holds even if $n$ is replaced by $\omega$.

\begin{theorem}[\flink{Theorem-15}]\label{theorem:RUDec}
	For any ordinal $\beta$, a proposition $P$ is $\beta$-decidable if and only if it is $\upa \beta$-decidable.
\end{theorem}
\begin{proof}
	If $\beta$ is a non-successor, there is nothing to do.
	Using \cref{item:spit-ordinal}, in the final case $\beta = \lambda + n$ for a limit (or zero) $\lambda$ and a finite positive $n$, and by \cref{theorem:SucOfSucIsOrdDecEquiv}, we may assume $n = 1$.
	Therefore, our goal is to show:
	\begin{equation*}
		\Big( \exists u. \,\,P\leftrightarrow\lambda+1 \leq u \Big) \longleftrightarrow
		\Big( \exists v. \,\,P\leftrightarrow \lambda + \omega \leq v \Big)
	\end{equation*}
	From left to right, let us assume we are given $u$; using \cref{lemma:RU23}, we have ${\lambda + 1 \leq u} \longleftrightarrow {\lambda + 1 \leq \upa u} \longleftrightarrow {\upa (\lambda + 1) \leq \upa u} \longleftrightarrow {\lambda + \omega \leq \upa u}$, so we choose $v \defeq \upa u$.
	From right to left, assume we are given $v$; we have ${\lambda + \omega \leq v} \longleftrightarrow {\lambda + \omega \leq \upa v} \longleftrightarrow {\lambda + 1 \leq \upa v}$, so we (again) choose $u \defeq \upa v$.
\end{proof}

As it turns out, we can phrase ordinal decidability purely in terms of non-successors.
\begin{definition}[\flink{Definition-16} Limit-decidability]
	For a proposition $P$ and $\lambda : \Brwzl$, we say that $P$ is $\lambda$-limit-decidable if
        $\exists y : \Brwzl. (P \leftrightarrow \lambda \leq y)$.
\end{definition}

\begin{theorem}[\flink{Theorem-17}]\label{thm:limit-decidable-suffices}
	For a proposition $P$ and an ordinal $\beta$, we have that $P$ is $\beta$-decidable if and only if it is $\upa \beta$-limit-decidable.
\end{theorem}
\begin{proof}
	Use \cref{theorem:RUDec}, then, in the non-trivial direction, round down the witness and use \cref{lemma:RU23}\cref{item:up-2}.
\end{proof}

\section{Closure under Binary Conjunctions}\label{sec:BinConj}

It is well-known that, if $P$ and $Q$ are decidable (or semidecidable), then so is $P \land Q$.
Spelling out the generalization of this statement to $\alpha$-decidability, it says that
if there exists $x, y : \Brw$ such that $P \leftrightarrow \alpha \leq x$ and $Q \leftrightarrow \alpha \leq y$, then there exists $z$ such that ${P \land Q} \leftrightarrow \alpha \leq z$.
This is in particular the case if $z$ is the meet of $x$ and $y$ when viewing $(\Brw,\leq)$ as a poset.
To avoid confusion with the poset of propositions (or the category of sets), and to support the intuition of ordinals as numbers, we refer to the meet of ordinals as their \emph{minimum}:

\begin{definition}[\flink{Definition-18} Minimum]
	We say that $\mu : \Brw$ is the \emph{minimum} of $\alpha$, $\beta : \Brw$ if the property
	\begin{equation}\label{eq:minimum-property}
		\forall \gamma. (\alpha \geq \gamma \land \beta \geq \gamma) \leftrightarrow \mu \geq \gamma
	\end{equation}
	is satisfied.
\end{definition}
If the minimum of $\alpha$ and $\beta$ exists, then it is necessarily unique;
thus, stating that all minima exist is equivalent to giving a function that calculates them.
The current section is structured as follows:
In \cref{subsec:minima-non-constructive}, we show that $\Brw$ cannot constructively have binary minima in general; in \cref{subsec:limmin-constructive}, we show that binary minima exist for non-successors; and consequently, in \cref{subsec:conjunections}, 
that $\alpha$-decidable propositions are closed under binary conjunctions.

\subsection{Non-Constructability of Binary Minima}\label{subsec:minima-non-constructive}

We cannot expect to be able to constructively calculate general binary minima:
if we could calculate the minimum of $\beta+1$ and $\blim f$, then we could check whether this minimum is a successor or a limit, and thus decide if $\beta$ is smaller than $\blim f$ or not --- but we know that $\beta < \blim f$ is not decidable in general.
The precise argument that no such minimum function can be constructed is that it would imply $\LPO$, the \emph{limited principle of omniscience}.
Recall that $\LPO$ states that the existential quantification of any decidable proposition is again decidable, or, in other words, every semidecidable proposition is decidable: in any binary sequence, there either exists a $1$, or the sequence is constantly $0$.
This is a constructive taboo which is not derivable in type theory, so anything implying it cannot be derivable either.

\begin{proposition}[\flink{Proposition-19}]\label{prop:appendix-minimum-nonconstructive}
	Assume that, for all ordinals $\alpha$ and $\beta$, there exists a minimum.
	Then, $\LPO$ holds.
\end{proposition}

\begin{proof}
	Let $s$ be a binary sequence; we aim to show that $s$ is either constantly $0$ or it contains the digit $1$.

	We choose $\alpha \defeq \blim {\jump s}$ and $\beta \defeq \omega + 1$, and assume that $\mu$ with the property \eqref{eq:minimum-property} exists.
	If $\mu = 0$, then $\gamma = 1$ gives an immediate contradiction.
	If $\mu = \mu' + 1$, we consider the case $\gamma \defeq \omega + 1$ and get that $\blim {\jump s} \geq \omega + 1$ (which is equivalent to $s$ containing the digit $1$) is logically equivalent to $\mu' \geq \omega$, which is decidable.
	Finally, if $\mu = \blim t$ for some sequence~$t$, setting $\gamma$ to $\blim t$ implies $\omega + 1 \geq \blim t$ and thus $\mu = \omega$. Therefore, the case $\gamma \defeq \omega + 1$ shows that $\blim {\jump s} \geq \omega + 1$ leads to a contradiction, which forces $s$ to be constantly $0$.
\end{proof}

\subsection{Construction of Binary Minima for Limits}\label{subsec:limmin-constructive}

While the above \cref{prop:appendix-minimum-nonconstructive}
seems like bad news for the conjunction of $\alpha$-decidable propositions,
what helps us is that general minima are not required.
Using \cref{theorem:RUDec} and \cref{thm:limit-decidable-suffices}, it suffices to define the minimum for zero and limits.

\begin{theorem}[\flink{Theorem-20}]\label{thm:limmin-exists}
	We have binary minima for non-successors: For $\alpha, \beta : \Brwzl$, there is $\limMiF \alpha \beta: \Brwzl$ satisfying \eqref{eq:minimum-property}.
\end{theorem}

Due to the quotient-inductive-inductive construction of $\Brw$,
a rigorous proof of this theorem is more involved than it may seem on paper.
The goal of the current subsection is a sketch of our proof.
Additional type-theoretic difficulties and our solutions are discussed in \cref{sec:Agda-technical-discussion}.
The complete proof is contained in our Agda formalization (linked to above).

The basic strategy is to define the ``minimum'' function on all pairs of ordinals, i.e., as a function $\limMin : \Brw \to \Brw \to \Brw$. The function is correct (calculates minima) if the inputs are non-successors, but necessarily incorrect in general.
The point constructors of the function are given as
\begin{alignat}{2}
	& \limMiF \bzero \beta &&= \bzero \notag \\
	& \limMiF {(\bsuc \alpha)} \beta &&= \limMiF \alpha \beta \notag \\
	& \limMiF {(\blim f)} \bzero &&= \bzero \notag \\
	& \limMiF {(\blim f)} {(\bsuc \beta)} &&= \limMiF {(\blim f)} \beta \label{eq:item-blim-bsuc}\\
	& \limMiF {(\blim \,f )}{(\blim \,g )} &&= \blim\, (\lambda \,n. \, (\limMiF{f_n}{g_n}) + n), \label{eq:item-blimblim}
\end{alignat}
which, in order to be accepted by Agda, can be expressed in a relational way (cf.~\cref{sec:Agda-technical-discussion}).
We omit the proof that \eqref{eq:item-blimblim} does not distinguish between bisimilar sequences and refer to the formalization.
More interesting is the monotonicity of $\limMin$ in both arguments, which is not only key for the well-definedness of \eqref{eq:item-blimblim}, but also for the proof of \cref{thm:limmin-exists}:

\begin{lemma}[\flink{Lemma-21}]\label{lem:limmin-mono}
	For all ordinals $\alpha_1 \leq \alpha_2$ and $\beta_1 \leq \beta_2$, we have $\limMiF {\alpha_1} {\beta_1} \leq \limMiF {\alpha_2} {\beta_2}$.
\end{lemma}
\begin{proof}
	We do induction on the four ordinals.
	The cases in which one or multiple are $\bzero$ or successor are easy; for example, if $\blim f \equiv \alpha_1 \leq \alpha_2 \equiv \bsuc \gamma$, the characterization of $\leq$ implies $\blim f \leq \gamma$ (cf.~\cref{lem:recall-basic-Brw-properties}\cref{lem:lim-under-suc}), a case that is covered by the induction hypothesis.
	The trickiest case occurs when all involved ordinals are limits.
    Therefore, let us assume that $\alpha_1 \equiv \blim f^1$, $\alpha_2 \equiv \blim f^2$, $\beta_1 \equiv \blim g^1$, and $\beta_2 \equiv \blim g^2$.
    The assumption now is that $f^1$~and~$g^1$ are simulated by $f^2$~and~$g^2$, respectively, and we have to show that the sequence $\lambda n. \limMiF {f^1_n} {g^1_n} + n$ is simulated by the sequence $\lambda m. \limMiF {f^2_m} {g^2_m} + m$. Therefore, let us fix $n$; from the assumption, we get $n_1$ such that $f^1_n \leq f^2_{n_1}$, and we get $n_2$ such that $g^1_n \leq g^2_{n_2}$.
    We claim that $m \defeq \max(n,n_1,n_2)$ fulfills the requirement.
    We have $f^1_n \leq f^2_{n_1} \leq f^2_m$ as well as $g^1_n \leq g^2_{n_2} \leq g^2_m$. By the induction hypothesis, together with $n \leq m$, we get
    $\limMiF {f^1_n} {g^1_n} + n\leq \limMiF {f^2_m} {g^2_m} + m$.
\end{proof}

The next major step in the construction is to establish the following properties:
\begin{lemma}[\flink{Lemma-22}]\label{lem:limaux-props}
	For all $\alpha, \beta : \Brw$, we have:
	\begin{romanenumerate}
		\item $\limMiF \alpha \beta \leq \alpha$, \label{item:limmin-first-aux}
		\item $\limMiF \alpha \beta \leq \beta$, \label{item:limmin-second-aux}
		\item $\limMiF \alpha \alpha = \dwa \alpha$. \label{item:limmin-third-aux}
	\end{romanenumerate}
\end{lemma}
\begin{proof}
	The property \eqref{item:limmin-first-aux} is shown by an induction on $\alpha$ and $\beta$, where the only interesting case is when both arguments are limits.
	Thus, we want to prove $\limMiF {(\blim f)} {(\blim g)} \leq \blim f$, i.e., for every $n : \N$, there exists $m : \N$ such that $(\limMiF {f_n} {g_n}) + n \leq f_m$.
	By the induction hypothesis, the left-hand side is below $f_n + n$, and choosing $m$ to be $n + n$ fulfills the requirement.

	\eqref{item:limmin-second-aux} is analogous.

	Due to the above properties, for \eqref{item:limmin-third-aux} it suffices to show $\dwa \alpha \leq \limMiF \alpha \alpha$, which we do by induction on $\alpha$.
	The least immediate case is $\alpha \equiv \blim f$.
	We need to show that $f$ is simulated by the sequence $h \defeq \lambda m. (\limMiF {f_m} {f_m}) + m$.
	Let us fix $n : \N$, and write $f_n$ as the sum of a non-successor and a natural number as given by \cref{item:spit-ordinal}, $f_n = \dwa f_n + k$.
	We then have
 	\begin{align*}
		f_n = \dwa f_n + k
		&\leq \dwa f_{n+k} + (n + k) \\
		&= \limMiF {f_{n+k}} {f_{n+k}} + (n+k)
		= h_{n+k}. \qedhere
	\end{align*}
\end{proof}

This allows us to show that we indeed have binary minima restricted to non-successors:

\begin{proof}[Proof of \cref{thm:limmin-exists}]
	We need to show that $\limMin$ calculates minima for non-successor arguments $\alpha$ and $\beta$, i.e., for every $\gamma : \Brw$, we need to prove
	\begin{equation*}
	(\gamma \leq \alpha \land \gamma \leq \beta) \longleftrightarrow (\gamma \leq \limMiF{\alpha}{\beta}).
	\end{equation*}
	The direction ``$\leftarrow$'' is a direct consequence of \cref{lem:limaux-props} \eqref{item:limmin-first-aux} and \eqref{item:limmin-second-aux}.
	For the direction ``$\rightarrow$,'' note that we can restrict ourselves to the case that $\gamma$ is zero or a limit due to
	\cref{lemma:RU23}\cref{item:up-1}.
	The assumption, together with \cref{lem:limmin-mono}, then implies $\limMiF{\gamma}{\gamma} \leq \limMiF{\alpha}{\beta}$; this suffices by \eqref{item:limmin-third-aux} of the previous lemma.
\end{proof}

\subsection{Binary Conjunctions}\label{subsec:conjunections}

The construction of minima makes it easy to show that $\alpha$-decidable propositions are closed under binary conjunctions:
\begin{theorem}[\flink{Theorem-23}]\label{thm:conjclosuremain}
	Let $\alpha : \Brw$. If $P$ and $Q$ are $\alpha$-decidable propositions, then so is $P \land Q$.
\end{theorem}
\begin{proof}
	As we are proving a proposition, we may assume that we are given $x$~and~$y$ with $P \leftrightarrow \alpha \leq x$ and $Q \leftrightarrow \alpha \leq y$.
	By \cref{thm:limit-decidable-suffices}, we can assume that $\alpha$,~$x$,~and~$y$ are all limit ordinals (or zero).
	\cref{thm:limmin-exists} now yields $P \land Q \leftrightarrow {\alpha \leq \limMiF{x}{y}}$.
\end{proof}

\section{Small Closures under Binary Disjunctions}\label{sec:bin-disj}

Having shown that $\alpha$-decidable propositions are closed under binary conjunction, the next step is to study whether they are closed under binary disjunction.
This is the case for decidable and for semidecidable propositions (i.e., for $\alpha = \omega$ and $\alpha = \omega \cdot 2$), but can unfortunately not be expected in full generality.
The reason is that no function with a property dual to the one of $\limMin$ is definable.
To make precise what is and what is not possible, we define the following refined version of binary maxima:
\begin{definition}[\flink{Definition-24} Maximum]\label{def:bin-max}
	Given $\alpha, \beta : \Brw$, and a property $P: \Brw \to \Prop$ of ordinals,
        an ordinal $\nu : \Brw$ is the maximum of $\alpha$ and $\beta$ with respect to ordinals satisfying~$P$ if:
	\begin{equation*}
		\forall \gamma. P(\gamma) \to \left((\alpha \geq \gamma \lor \beta \geq \gamma) \leftrightarrow \nu \geq \gamma\right).
	\end{equation*}
\end{definition}

We show in \cref{sec:no-bin-max} that binary maxima cannot exist globally even for non-successors.
However, these maxima exist with respect to sufficiently small $\alpha$ (\cref{sec:binmax}), which means that such $\alpha$-decidable propositions are closed under binary disjunction (\cref{sec:small-bin-disj}).

\subsection{Non-Constructability of Binary Maxima}\label{sec:no-bin-max}

From the assumption that binary maxima exist for arbitrary ordinals, we could derive a similar taboo as in \cref{prop:appendix-minimum-nonconstructive}; we do not demonstrate this here.
What is more interesting is that, unlike for minima, even restricting ourselves to non-successors cannot work:

\begin{proposition}[\flink{Proposition-25}]\label{prop:maximum-nonconstructive}
	If binary maxima exist for non-successor ordinals, then $\LPO$ holds.
\end{proposition}

The proof relies on the following simple observation,
which extends known connections between properties of $\Brw$ and $\LPO$ \cite[Thms 48 and 53]{kraus2023ordinals}.

\begin{lemma}[\flink{Lemma-26}]\label{lem:linearity-to-lpo}
	The following are equivalent:
	\begin{romanenumerate}
		\item\label{item:linear}
		The order $\leq$ is linear (i.e., for all ordinals $\alpha$ and $\beta$, we have $\alpha \leq \beta \lor \beta \leq \alpha$).
		\item\label{item:nonsuc-linear}
		The order $\leq$ is linear for non-successors (i.e., if $\alpha$ and $\beta$ are non-successors, then $\alpha \leq \beta \lor \beta \leq \alpha$).
		\item\label{item:LPO}
		$\LPO$.
	\end{romanenumerate}
\end{lemma}
\begin{proof}
	The implication ``\cref{item:linear}$\to$\cref{item:nonsuc-linear}'' is trivial.
    For ``\cref{item:LPO}$\to$\cref{item:linear}'' we note that, by \cite[Thm~53]{kraus2023ordinals}, $\LPO$ implies trichotomy ($(\alpha < \beta) \lor (\alpha = \beta) \lor (\alpha > \beta)$), which implies linearity.

	For the last implication ``\cref{item:nonsuc-linear}$\to$\cref{item:LPO}'',
	let $f : \N \to \Bool$ be a binary sequence.
	In \cref{def:jumping}, we defined the sequence~$\jump f$ that ``jumps'' when a $1$ occurs in $f$;
	here, we want to ``double jump''
	by replacing $\omega$ by $\omega \cdot 2$ in \cref{def:jumping}.
	We refer to the resulting sequence as $\doublejump f$.

	Assuming linearity, we get ${(\blim {\doublejump f} \leq \omega \cdot 2)} \lor {(\omega \cdot 2 \leq \blim {\doublejump f})}$.
	If $\blim {\doublejump f} \leq \omega \cdot 2$, then $f_k = 1$ leads to a contradiction, so the sequence $f$ is constantly $0$.
	If $\omega \cdot 2 \leq \blim {\doublejump f}$, then, by \cref{lem:recall-basic-Brw-properties}\cref{lem:x-under-limit-means-x-under-element}, we have $\exists n.\, \omega \leq {\doublejump f}_n$. In this case, the values $f_0, \ldots, f_{n-1}$ cannot all be~$0$, so we can find a $1$ by checking these finitely many binary values under the truncation monad.
\end{proof}

\begin{proof}[Proof of \cref{prop:maximum-nonconstructive}]
  As binary maxima are
  unique as soon as they exist, the assumption yields a function ${m : \Brwzl \to \Brwzl \to \Brw}$ calculating maxima with trivial condition ${P \equiv \lambda \delta.\Unit}$.
	Fixing $\alpha$ and $\beta$ in \cref{def:bin-max} with $\nu = m(\alpha, \beta)$, and setting $\gamma$ first to $\alpha$, then to~$\beta$, and finally to $m(\alpha,\beta)$ gives $\alpha \leq m(\alpha,\beta)$, $\beta \leq m(\alpha,\beta)$, and $(m(\alpha,\beta) \leq \alpha) \lor (m(\alpha,\beta) \leq \beta)$. Putting these together, we get $\beta \leq \alpha \lor \alpha\leq \beta$, and \cref{lem:linearity-to-lpo} implies $\LPO$.
\end{proof}

\subsection{Binary Maxima of Small Ordinals}\label{sec:binmax}

Notwithstanding the above result,
we can calculate a function that has the universal property of the maximum with respect to ordinals of shape $\omega \cdot k$:

\begin{theorem}[\flink{Theorem-27}]\label{thm:some-max-exist}
	Binary maxima of non-successor ordinals exist with respect to ordinals of the form $\omega \cdot k$.
        That is, we can construct a function $\limMax : \Brw \to \Brw \to \Brw$
	such that, for all non-successors $\alpha$,$\,\beta$ and for all $\gamma$ such that $\exists k: \N.\, \gamma = \omega \cdot k$,
	we have
	\begin{equation}\label{eq:max-prop}
		(\alpha \geq \gamma \lor \beta \geq \gamma) \longleftrightarrow \limMaF \alpha \beta \geq \gamma.
	\end{equation}
\end{theorem}
\begin{proof}[Proof sketch]
	Much of the construction of $\limMax$ is analogous to that of $\limMin$; it is given in detail in our formalization, and our approach to a technical difficulty related to the usage of a QIIT is discussed in \cref{sec:Agda-technical-discussion}.
	The defining equations are
	\begin{alignat*}{2}
		& \limMaF \bzero \beta &&\defeq \beta \\
		& \limMaF {(\bsuc \alpha)} \beta &&\defeq \limMaF \alpha \beta \\
		& \limMaF {(\blim f)} \bzero &&\defeq \blim f \\
		& \limMaF {(\blim f)} {(\bsuc \beta)} &&\defeq \limMaF{(\blim f)}{\beta} \\
		& \limMaF {(\blim f)} {(\blim g)} &&\defeq  \blim \left(\lambda \,n. \, (\limMaF {f_n} {g_n}) + n\right). \label{eq:item:limmax-blimblim}
	\end{alignat*}

	After we have defined $\limMax$ satisfying the equations above, we need to show that it satisfies \eqref{eq:max-prop}.
	For the ``$\rightarrow$'' direction, it suffices to show $\alpha \leq \limMaF \alpha \beta$ and $\beta \leq \limMaF \alpha \beta$, which is easy to prove by induction on $\alpha$ and $\beta$.

	In the assumption $\exists k: \N. \, \gamma = \omega \cdot k$, the natural number $k$ is unique,
	and we prove the ``$\leftarrow$'' direction by induction on $k$.
	There is nothing to do for $k = 0$.
	Next, we assume that $\limMaF \alpha \beta \geq \omega \cdot (k+1)$.
	If one (or both) of $\alpha$, $\beta$ are $\bzero$, the conclusion is trivial; otherwise, assume $\alpha = \blim f$ and $\beta = \blim g$.
	By definition of $\limMax$ and the property \cref{lem:recall-basic-Brw-properties}\cref{lem:x-under-limit-means-x-under-element}, the assumption implies that there exists an $m : \N$ such that $\limMaF {f_m} {g_m} + m \geq \omega \cdot k$, and by property \cref{lem:recall-basic-Brw-properties}\cref{lem:lim-under-suc}, we get $\limMaF {f_m} {g_m} \geq \omega \cdot k$.
	The induction hypothesis implies
	\begin{equation*}
		({f_m} \geq \omega \cdot k) \lor ({g_m} \geq \omega \cdot k).
	\end{equation*}
	Because $\blim h \geq h_m + \omega$ holds for any strictly increasing sequence, this gives
	\[
          ({\blim f} \geq \omega \cdot (k+1)) \lor ({\blim g} \geq \omega \cdot (k+1)).
          \qedhere
	\]
\end{proof}

\subsection{Binary Disjunction}\label{sec:small-bin-disj}

Although $\limMax$ cannot prove general closure of $\alpha$-decidability under binary disjunction by \cref{prop:maximum-nonconstructive}, it can prove it for ordinals of the form $\alpha = \omega \cdot k + n$.
For $\alpha = \omega \cdot 1$ and $\alpha = \omega \cdot 2$, this is the usual known closure under disjunction of decidable and semidecidable propositions respectively.

\begin{theorem}[\flink{Theorem-28}]\label{thm:disjclosureunderomega^2}
	For natural numbers $k$ and $n$, the $(\omega \cdot k + n)$-decidable propositions are closed under binary disjunction.
\end{theorem}
\begin{proof}
	By \cref{theorem:RUDec}, we may assume $n = 0$.
	If $P \leftrightarrow \omega \cdot k \leq x$ and $Q \leftrightarrow \omega \cdot k \leq y$, then
        we have $P \lor Q \leftrightarrow \omega \cdot k \leq \limMaF{x}{y}$
        by \cref{thm:some-max-exist}.
\end{proof}

\section{Semidecidable Families and Quantifiers}\label{sec:semidec-and-quantifiers}

Having studied binary meets and joins of $\alpha$-decidable propositions, we get back to the situation discussed in the introduction.
Assume we are given a family $P : \N \to \Prop$ of semidecidable propositions, for example the twin prime family ($n \mapsto $ ``there exists a twin prime pair above $n$'').
In \cref{subsec:countable-meets}, we show that its infinite meet $\forall i. P_i$, which for the twin prime family corresponds to the twin prime conjecture, is $\omega^2$-decidable.
In \cref{subsec:countable-joins}, we show that its join $\exists i. P_i$ is $(\omega \cdot 3)$-decidable, before presenting results for quantifier alternations in \cref{subsec:quantifier-alternations}.

All these results critically depend on the construction of the \emph{characteristic ordinal} of such a family $P$ that we explain in \cref{subsec:characteristic-ordinal}.
\cref{subsec:normalizing-families} contains the simple but useful observation that such families can often be assumed to be increasing or decreasing.

\subsection{Characteristic Ordinals of Families of Semidecidable Propositions}
\label{subsec:characteristic-ordinal}

Similar to how we combined the information given by two $\alpha$-decidable propositions in order to calculate a new proposition in \cref{sec:BinConj,sec:bin-disj}, we now need to express the information contained in a countable family $P$ of propositions with a single ordinal.
We restrict ourselves to the case that every $P_n$ is semidecidable.
A discussion of the more general case of $\alpha$-decidability can be found in \cref{sec:conclusions}.

Concretely, given a sequence $P : \N\to \Prop$ of semidecidable propositions, we want to construct a single ordinal $\Psi(P)$ that, in a certain sense, measures how many of the propositions are true --- we call it the \emph{characteristic ordinal} $\Psi(P)$ of $P$.
It is approximated by $\Psi_n(P)$, an ordinal that measures how many of the propositions $P_0, P_1, \ldots, P_{n-1}$ are true.
If these were assumed to be decidable, we could simply define $\Psi_n(P)$ to be this number, but, as the $P_i$ are only assumed to be semidecidable,
the definition of $\Psi_n(P)$ is more involved.
Afterwards, we can take the limit of a construction involving $\Psi_n(P)$ in order to define $\Psi(P)$.

Our strategy for the first part is as follows.
Assume we are given $n : \N$, together with a set
$\{P_0, P_1, \dots, P_{n-1}\}$ of propositions. We want to construct a function
\begin{equation*}
	\Psi_n : \left(\Pi_{i \in \{0,\ldots,n-1\}} \isSemiDec(P_i)\right) \to \Brw.
\end{equation*}
To do this, we take several steps:
\begin{enumerate}
	\item First, we construct a similar function assuming we have access to the data contained in the semidecidability proofs.
	Recall that $\isSemiDec(Q)$ is defined to mean
	\begin{align*}
		&\isSemiDec(Q) \defeq \exists s: \N \to \Bool. \, Q \leftrightarrow \exists i. s_i = 1; \\
		\intertext{by its \emph{data}, we mean the non-truncated version of the same type,}
		&\SemiDec(Q) \defeq \Sigma s: \N \to \Bool. \, Q \leftrightarrow \exists i. s_i = 1.
	\end{align*}
	Concretely, we construct a function
	\begin{equation*}
		\Psistruc_n: \left(\Pi_{i \in \{0,\ldots,n-1\}} \SemiDec(P_i)\right) \to \Brw.
	\end{equation*}
	\item In the second step, we then prove that $\Psistruc_n$ is weakly constant, meaning that it maps any pair of inputs to equal outputs.
	As the codomain is a set, this implies that $\Psistruc_n$ factors through the propositional truncation~\cite{kraus2013hedberg}.
	We therefore get:
	\begin{equation*}
		\Psitrunc_n : \left\| \Pi_{i \in \{0,\ldots,n-1\}} \SemiDec(P_i)\right\| \to \Brw.
	\end{equation*}
	\item Using \emph{finite choice}, i.e., the $n$-fold application of the equivalence $\| A \times B \| \simeq \|A\| \times \|B\|$, we get the equivalence
	\[
          \FChoice_n : \left(\Pi_{i \in \{0,\ldots,n-1\}} \isSemiDec(P_i)\right) \to \left\| \Pi_{i \in \{0,\ldots,n-1\}} \SemiDec(P_i)\right\|.
	\]
	Then, $\Psi_n$ is the composition of $\Psitrunc_n$ with $\FChoice_n$.
\end{enumerate}

The details for the first two steps are as follows.

\subsubsection*{Construction of $\Psistruc_n$:}

Suppose, for each $i<n$, we are given a Boolean sequence $s^i : \N \to \Bool$ that semidecides $P_i$.
We replace each sequence~$s^i$ by an increasing sequence $t^i$ that is defined by $t^i_0 \defeq s^i_0$ and $t^i_{n+1} \defeq t^i_n \lor s^i_{n+1}$; in other words, $t^i$ is the ``upwards normalization'' of $s^i$.
Let $t : \N \to \N$ be the pointwise sum of the $t^i$; i.e., $t_k \defeq t^0_k + t^1_k + \ldots + t^{n-1}_k$, where we implicitly convert $\bff$ to 0 and $\btt$ to 1.

The construction ensures that the sequence $t$ is (weakly) increasing. We set
\begin{equation}\label{eq:construction-of-Psi}
	\Psistruc_n := \,\blim (\lambda k. \,\omega \cdot t_k + k),
\end{equation}
where the final ``${\!} + k$'' ensures that the constructor takes a \emph{strictly} increasing sequence as an argument.

\subsubsection*{Weak constancy of $\Psistruc_n$:}

For the argument that $\Psistruc_n$ is weakly constant, the following observation is useful:
\begin{lemma}[\flink{Lemma-29}]\label{lemma:boolsim}
	If $(s,q)$ and $(s',q')$ are 
        semidecidability data for the same proposition $P$, i.e., terms of the type $\SemiDec (P)$, then $s$~and~$s'$ are bisimilar, i.e., we have $\forall k.\, \exists m.\, s_k \leq s'_m$ and $\forall k.\, \exists m.\, s'_k \leq s_m$.
\end{lemma}
\begin{proof}
	By symmetry, it suffices to show that $s$ is simulated by~$s'$.
	Let $k : \N$ be given.
	If $s_k = 0$, the statement is trivial.
	If $s_k = 1$,
	then $P$ holds by $q$.
	From $q'$, we therefore get $\exists m. {s'}_m = 1$.
\end{proof}

\begin{lemma}[\flink{Lemma-30}]
  The function $\Psistruc_n$ is weakly constant, i.e.,
  $\forall s,\overline{s}. \, \Psistruc_n(s) = \Psistruc_n(\overline{s})$.
\end{lemma}
\begin{proof}
	Assume that, for each $i < n$, we have two sets of sequences $s^i,{\overline s}^i : \N \to \Bool$ that both semidecide $P_i$.
	Let $t^i$, ${\overline t}^i$ be their respective normalized-up sequences. According to the second step of the previous construction, this defines functions $t, \overline t: \N \to \{0,1,\ldots,n\} $, which gives us sequences $T, \overline T : \N \to \Brw$ defined by $T_k \defeq \omega \cdot t_k + k$ and  $\overline T_k \defeq \omega \cdot \overline t_k + k$. To prove that $\Psistruc_n$ is weakly constant, we need to show that $\overline T$ simulates $T$ (and the symmetric statement, which is analogous).
	Therefore, our goal is to show:
	\begin{equation*}
		\forall k. \, \exists m. \,\omega\cdot t_k + k \leq \omega \cdot \overline t_m + m
	\end{equation*}
	For this, it suffices to show
	$	\forall k. \, \exists m. \,  t_k \leq  \overline t_m$,
	which is straightforward using \cref{lemma:boolsim}.
\end{proof}

This completes the construction of $\Psi_n$.

\begin{definition}[\flink{Definition-31} Characteristic ordinal of a sequence]
	Given a sequence of semidecidable propositions $P : \N \to \Prop$,
	the above construction gives us a sequence $n \mapsto \psiapprox n P$, where $\psiapprox n P : \Brw$ is the function applied to the assumed proof of semidecidability.
	This sequence is (weakly) increasing by construction.
	We define the \emph{characteristic ordinal of $P$}, written $\Psi(P)$, to be
	\begin{equation*}
		\Psi(P) \defeq \blim (\lambda n.\, \psiapprox n P + n).
	\end{equation*}
\end{definition}

\subsection{Normalizing Families}\label{subsec:normalizing-families}

The following observation tells us that we can assume that the family $P$ is downwards (or upwards) closed, which will be a useful technical device:

\begin{lemma}[\flink{Lemma-32-i}]\label{downward-closed-is-not-harmful}
	Let $P : \N \to \Prop$ be a family of propositions.
	\begin{romanenumerate}
		\item Let $\alpha : \Brw$. If each $P_i$ is $\alpha$-decidable, then we can construct a family $Q : \N \to \Prop$ of $\alpha$-decidable propositions such that $\forall k. Q_{k+1} \to Q_k$, and such that $\forall k. P_k \leftrightarrow \forall k. Q_k$.
		\item Let $k$, $n$ be natural numbers. If each $P_i$ is $(\omega \cdot k + n)$-decidable,
		then we can construct a family $Q : \N\to \Prop$ of $(\omega \cdot k + n)$-decidable propositions such that $\forall i. (Q_i \to Q_{i+1})$ and $\exists k. P_k \leftrightarrow \exists k. Q_k$.
	\end{romanenumerate}
\end{lemma}
\begin{proof}
	In the first case, we define $Q$ by $Q_i \defeq \forall_{j\leq i} \, P_j$.
	This is a finite conjunction and thus $\alpha$-decidable by \cref{thm:conjclosuremain}.
	In the second case, we set $Q_i \defeq \exists_{j\leq i} \, P_j$, and apply \cref{thm:disjclosureunderomega^2}.
\end{proof}

\subsection{Countable Meets of Semidecidable Propositions}\label{subsec:countable-meets}

Our first application of the characteristic ordinal $\Psi(P)$ is the case of countable meets:
\begin{theorem}[\flink{Theorem-33}]\label{thm:forallPnomega2}
	Let $P: \N \to \Prop$ be a family of semidecidable propositions.
	Then $\forall n. \, P_n$ is $\omega^2$-decidable.
\end{theorem}
To give an example, the twin prime conjecture (cf.~the introduction) is an instance of this theorem, and thus $\omega^2$-decidable.

In order to prove \cref{thm:forallPnomega2}, we need several additional observations about the characteristic ordinal $\Psi(P)$ and its approximations $\psiapprox n P$.
For the rest of this subsection, assume $P: \N \to \Prop$ is a given family of semidecidable propositions.

\begin{lemma}[\flink{Lemma-34} Key Lemma] \label{lemma:secondKeyLemma}
	We have that
	$\exists n. \,\omega \cdot 2 \leq \psiapprox n P$ if and only if $\exists m. \,P_m$.
\end{lemma}
\begin{proof}
	The following arguments rely heavily on the fact that the respective goals are propositions, which allows us to replace $\exists$ by $\Sigma$ several (but finitely many) times.

	We first prove the left to right direction. We can assume that we are given an $n$ such that $\omega \cdot 2 \leq \psiapprox{n}{P}$.
	If we unfold the latter expression using the notation and construction of \cref{subsec:characteristic-ordinal}, and especially recall \eqref{eq:construction-of-Psi},
	this assumption becomes
	\begin{equation*}
		\omega \cdot 2 \leq \blim (\lambda k. \omega \cdot t_k + k).
	\end{equation*}
	Using \cref{lem:recall-basic-Brw-properties}\cref{lem:lim-under-lim-simulation}, this implies that there is $k$ such that $\omega + 1 \leq \omega \cdot t_k + k$, which implies $t_k \geq 1$.
	This allows us to check what the least $m$ is such that $t_m \geq 1$, and the construction ensures that $P_m$ holds.

	For the right to left direction,
	assume we have $m$ such that $P_m$ holds.
	This implies that the sequence witnessing the semidecidability of $P_m$ contains the digit $1$, and unfolding the construction, this implies
	$\psiapprox {m+1} P \geq \omega \cdot 2$.
\end{proof}

For a given natural number $k$, the sequence $\lambda i. \, P_{i + k}$ is the sequence $P$ without the first $k$ entries.
By construction, we trivially have $\Psi(\lambda i. \, P_{i + k}) \leq \Psi(P)$.
The other direction does not hold in general, but we have the following:

\begin{lemma}[\flink{Lemma-35}] \label{lemma:cutting}
	For any natural number $k$, we have
	\begin{equation*}
		\Psi(P) \geq \omega^2
		 \; \rightarrow \;
		\Psi(\lambda i. \, P_{i + k}) \geq \omega^2.
	\end{equation*}
\end{lemma}
\begin{proof}
	It suffices to prove the lemma for $k=1$, as the general case follows by finite iteration.
	By construction, for all natural numbers $n$, we have
	\begin{equation}\label{eq:psi-shifted-inequality}
		\psiapprox{n}{P} \leq \psiapprox{n}{\lambda i. \, P_{i+1}}  + \omega.
	\end{equation}

	We assume $\Psi(P) \geq \omega^2$.
	For each $m$, by \cref{lem:recall-basic-Brw-properties}\cref{lem:lim-under-lim-simulation} applied to the index $m+1$, there exists an $n$ such that $\psiapprox{n}{P} + n \geq \omega \cdot (m+1)$.
	Combining this with \eqref{eq:psi-shifted-inequality} gives
	${{\psiapprox{n}{\lambda i. \, P_{i+1}}} + {\omega} + {n}} \geq {{\omega \cdot m} + \omega}$.
	Using \cref{lem:recall-basic-Brw-properties}\cref{lem:lim-under-suc} and cancellation of addition with $\omega$ (a consequence of \cref{lem:recall-basic-Brw-properties}\eqref{lem:lim-under-suc}--\eqref{lem:lim-under-lim-simulation}),
	we get $\psiapprox{n}{\lambda i. \, P_{i+1}} \geq \omega \cdot m$,
	which means that the sequence $m \mapsto \omega \cdot m$ is simulated by the sequence $n \mapsto \psiapprox{n}{\lambda i. \, P_{i+1}} + n$ as required.
\end{proof}

Next we observe that, if we know whether $P_0$ to $P_n$ hold, we can calculate explicitly what $\psiapprox n P$ is.
We only require the following special case:

\begin{lemma}[\flink{Lemma-36}]\label{lemma:firstKeyLemma}
	Assume that $P$ is downwards closed and $P_n$ holds. Then, we have $\psiapprox n P = \omega \cdot (n+1)$.
	\begin{proof}
		Since $P$ is downwards closed, all $P_i$ (with $i \leq n$) are true, and their corresponding sequences contain the digit $1$.
		Tracking the definition of $\Psi$, we get the claimed equation.
	\end{proof}
\end{lemma}

With the above lemmas at hand, we are ready to prove that countable meets of semidecidable propositions are $\omega^2$-decidable:

\begin{proof}[Proof of \cref{thm:forallPnomega2}]
By \cref{downward-closed-is-not-harmful}, we can assume that the sequence $P : \N \to \Prop$ is downwards closed, i.e., that $P_{k+1} \to P_k$ for each $k$.
We claim
\begin{equation}\label{eq:goal-of-inf-meet-thm}
	\forall i. P_i \longleftrightarrow \Psi(P) \geq \omega^2,
\end{equation}
which shows that $\forall i. P_i$ is $\omega^2$-decidable.

The direction ``$\to$'' is directly implied by \cref{lemma:firstKeyLemma}.
For the direction ``$\leftarrow$'', we assume $\Psi(P) \geq \omega^2$.
For an arbitrary natural number $n$, we need to show that $P_n$ holds.
By \cref{lemma:cutting},
we have $\Psi(\lambda i. \, P_{i+n}) \geq \omega^2$,
and \cref{lem:recall-basic-Brw-properties}\cref{lem:lim-under-lim-simulation} implies that there exists a $j$ such that
$\psiapprox{j}{\lambda i. \, P_{i+n}} + j \geq \omega \cdot 2$ and therefore, by \cref{lem:recall-basic-Brw-properties}\cref{lem:lim-under-suc},
$\psiapprox{j}{\lambda i. \, P_{i+n}} \geq \omega \cdot 2$.
\cref{lemma:secondKeyLemma} shows that there exists $l$ such that $P_{l+n}$ holds, and, as $P$ is downwards closed, this means that $P_n$ holds.
\end{proof}

Immediate consequences are the following:
\begin{corollary}[\flink{Corollary-37-1}]\label{cor:neg-of-semidec-is-omega2}
	The countable meet of a family of decidable propositions is $\omega^2$-decidable.
	The negation of a semidecidable proposition is also $\omega^2$-decidable.
\end{corollary}
\begin{proof}
	For the first statement, observe that every decidable proposition is also semidecidable; and for the second, note that the negation of a semidecidable proposition is true if and only if all digits in any defining sequence are $0$.
\end{proof}

\subsection{Countable Joins of Semidecidable Propositions}\label{subsec:countable-joins}

Having treated countable meets, the next case of interest is the one of countable joins, in the following form:

\begin{theorem}[\flink{Theorem-38}]\label{thm:existsPnomega2}
	Let $P: \N \to \Prop$ be a family of semidecidable propositions. Then, $\exists n. \, P_n$ is $(\omega \cdot 3)$-decidable.
\end{theorem}
\begin{proof}
	We show
	\begin{equation*}
		\Psi(P) \geq \omega \cdot 3 \, \longleftrightarrow \, \exists n. \, P_n.
	\end{equation*}

	Note that the left-hand side is equivalent to $\exists k. \, \psiapprox{k}{P} \geq \omega \cdot 2$: Any sequence containing an element above $\omega \cdot 2$ has a limit of at least $\omega \cdot 3$, and \cref{lem:recall-basic-Brw-properties}\cref{lem:lim-under-lim-simulation} together with \cref{lem:recall-basic-Brw-properties}\cref{lem:lim-under-suc} shows the other direction.
	The statement is therefore given by \cref{lemma:secondKeyLemma}.
\end{proof}

\subsection{Quantifier Alternations}\label{subsec:quantifier-alternations}

Having discussed the decidability of $\forall i. \, P_i$ and $\exists i. \, P_i$ for semidecidable families,
a natural next step is to study more general formulas of the arithmetical hierarchy, i.e., multiple quantifiers and quantifier alternations.
We do not know whether the general case has a satisfactory solution;
here, we only present a special case that generalizes the decidability question of the search for counter-examples of relatively complicated problems.

As a preparation, we record the following observation that straightforwardly follows from the construction of the characteristic ordinal in \cref{subsec:characteristic-ordinal}.
\begin{lemma}[\flink{Lemma-39}]\label{lem:essentiallimit1-inproofOfOmega^2}
	For two sequences of semidecidable propositions $P,Q : \N \to \Prop$ with $\forall n. \,P_n \to Q_n$, we have $\Psi(P) \leq \Psi(Q)$. \qed
\end{lemma}

This allows us to show:
\begin{theorem}[\flink{Theorem-40}]\label{thm:existsforallSemidec}
	If $P: \N \times \N \to \Prop$ is a family of semidecidable propositions such that, for all $m$ and $n$, we have $ P(n , m) \to P(n,m+1)$, then $\exists m. \, \forall n. \, P(n,m)$ is $(\omega^2 + \omega)$-decidable.
\end{theorem}
\begin{proof}
	For every $m : \N$ we have, by the construction of \eqref{eq:goal-of-inf-meet-thm}, that $\forall n. \, P(n,m)$ is logically equivalent to $\Psi(\lambda n. \, P(n,m)) \geq \omega^2$.
	By the upwards closure assumption of the theorem statement and \cref{lem:essentiallimit1-inproofOfOmega^2},
	we get that the sequence
	\begin{equation*}
		m \mapsto \Psi(\lambda n. \, P(n,m)) + m
	\end{equation*}
	is increasing, and the summand ``${\!} + m$'' ensures that it is \emph{strictly} increasing.
	If we write $\lambda$ for the limit of this sequence, it is easy to observe that
	\begin{equation*}
		\lambda \geq \omega^2 + \omega \, \longleftrightarrow \, \exists m. \, \forall n. \, P(n,m),
	\end{equation*}
	which proves the theorem.
\end{proof}

As an instance of this, we see that the existence of a counterexample for the twin prime conjecture is $(\omega^2 + \omega)$-decidable.
Indeed, 
writing $P(n,m)$ to mean that ${n \geq m}$ implies that not both the numbers $n$ and $n+2$ are prime, 
a counterexample to the twin prime conjecture consists of a number $m$ such that $\forall n. \, P(n,m)$. Note that the proposition $P(n,m)$ is decidable and therefore semidecidable, and that it is easy to check that $P(n,m) \to P(n,m+1)$ holds.

\section{Ordinal Decidability and Countable Choice}\label{sec:CC}

In the previous section, we showed in \cref{thm:existsPnomega2} that the countable join of semidecidable propositions is $(\omega \cdot 3)$-decidable.
Without further assumptions, it does not seem possible to show that this join is semidecidable, as such a closure property would imply \emph{Escard\'o-Knapp choice}~\cite{escardo2017partial,dejong2022semidecidability}.%
\footnote{This should not be confused with a related principle also called Escard\'o-Knapp choice (EKC) as introduced by Andrew Swan in two talks from 2019~\cite{Swan2019a,Swan2019b}.}
However, it is well-known that this closure property holds in the presence of \emph{countable choice}, a principle which says that, for any type family $B : \N \to \UU$, we have
\begin{equation*}
	\left(\forall n:\N. \, \trunc{B(n)} \right) \, \to \, \trunc{\Pi n:\N.\, B(n)};
\end{equation*}
in other words, that we can commute the propositional truncation with countable products.
We sketch the argument in \cref{lem:cc-implies-closure-under-joins} below.
In \cref{subsec:CC}, we present further consequences of countable choice for the hierarchy of decidability.
Afterwards, in \cref{subsec:Sierpinski}, we discuss \emph{Sierpi\'nski-semidecidability}, an alternative definition of semidecidability that is closed under countable joins without the assumption of a choice principle, and its connection with the notions studied in this paper.

\subsection{Consequences of Countable Choice}\label{subsec:CC}

The following is Theorem~4 of Escard\'o and Knapp~\cite{escardo2017partial}.
\begin{lemma}[\flink{Lemma-41}]\label{lem:cc-implies-closure-under-joins}
	Assuming countable choice, the countable join of semidecidable propositions is semidecidable.
\end{lemma}
\begin{proof}[Proof sketch]
	If there exists a binary sequence $s_n$ for every natural number $n$, then countable choice implies that there exists a sequence $\N \times \N \to \Bool$ that, when composed with any bijection $\N \cong \N \times \N$, gives a binary sequence that semidecides the infinite join.
\end{proof}
From the point of view of \cref{thm:existsPnomega2}, the above lemma says that certain $(\omega \cdot 3)$-decidable propositions are semidecidable.
This result can be strengthened: assuming countable choice, not only $(\omega \cdot 3)$-decidable, but in fact all $(\omega \cdot k)$-decidable propositions are semidecidable.

The following auxiliary statement follows directly by unfolding and substituting:
\begin{lemma}[\flink{Lemma-42}]\label{lem:larger-alpha-is-beta-dec}
	Let $\alpha, \beta : \Brw$ be given. If, for all ordinals $x$, the proposition $x \geq \alpha$ is $\beta$-decidable, then every $\alpha$-decidable proposition is $\beta$-decidable. \qed
\end{lemma}

This implies:
\begin{theorem}[\flink{Theorem-43}]\label{thm:CCimpliesOmega.kTobeSemidec}
	Let $k$ be a natural number. Assuming countable choice, every $(\omega \cdot k)$-decidable proposition is semidecidable.
\end{theorem}

\begin{proof}
	We show by induction on $k$ that, for all $x$, the proposition $x \geq \omega \cdot k$ is semidecidable,
	which implies the theorem by \cref{lem:larger-alpha-is-beta-dec}.
	The claimed property is trivial for $k = 0$ ($0$\nobreakdash-decidable propositions are true), $k = 1$ ($\omega$-decidable propositions are decidable), and $k = 2$ ($(\omega \cdot 2)$\nobreakdash-decidable propositions are semidecidable).
	Assume it holds for $k$.
	We show by induction on $x$ that $x \geq {{\omega \cdot k} + \omega}$ is semidecidable.
	The only interesting case is when $x$ is a limit, $x = \blim f$. We have:
        \[
		\blim f \geq \omega \cdot k + \omega
		\longleftrightarrow \forall n. \, \exists m. \, f_m \geq \omega \cdot k + n
		\longleftrightarrow \exists m. \, f_m \geq \omega \cdot k,
        \]
	where the first step holds by \cref{lem:recall-basic-Brw-properties}\cref{lem:lim-under-lim-simulation} and the second because, if $f_m \geq {\omega \cdot k}$, then $f_{m+n} \geq {f_m + n} \geq {{\omega \cdot k} + n}$ by strict monotonicity of $f$.
	The last line is the countable join of propositions that are semidecidable by the induction hypothesis, making it semidecidable by \cref{lem:cc-implies-closure-under-joins}.
\end{proof}

However, it is not the case that the whole hierarchy collapses, as $\omega^2$-decidable propositions are indeed distinguishable from $(\omega \cdot k)$-decidable ones.
Recall that \emph{Markov's principle}, written $\MP$, states that if the binary sequence $s : \N \to \Bool$ is not constantly $0$, then there exists an index $i$ such that $s_i = 1$.
This is strictly weaker than $\LPO$.
In particular, there are models where both countable choice and $\MP$ hold, but $\LPO$ does not~\cite{hendtlass_separating_LPO}, which means that the following result shows that countable choice is not sufficient to prove that $\omega^2$-decidable propositions are semidecidable.

\begin{theorem}[\flink{Theorem-44}]
	If every $\omega^2$-decidable proposition is semidecidable,
	then $\MP$ implies $\LPO$.
\end{theorem}
\begin{proof}
	$\LPO$ is equivalent to the statement that every semidecidable proposition is decidable; thus, assume $P$ is semidecidable.
	By \cref{cor:neg-of-semidec-is-omega2}, $\neg P$ is $\omega^2$-decidable, and thus, by the assumption, it is semidecidable. $\MP$ implies that a proposition is decidable when it and its negation are both semidecidable.
\end{proof}

Our next observation is that, under countable choice, every $\omega^2$\nobreakdash-decidable proposition can be assumed to be a countable meet of semidecidable propositions:
\begin{theorem}[\flink{Theorem-45}]
    Let $P$ be $\omega^2$-decidable. Assuming countable choice, there exists a family $Q : \N \to \Prop$ of semidecidable propositions such that $P \leftrightarrow \forall n. Q_n$.
\end{theorem}
\begin{proof}
    By assumption, there exists an $x$ such that $P \leftrightarrow x \geq \omega^2$, and by \cref{thm:limit-decidable-suffices}, we can assume that $x = \blim f$ for some $f$.
    By the characterisation in \cref{lem:recall-basic-Brw-properties}\cref{lem:lim-under-lim-simulation}, the expression $\blim f \geq \omega^2$ is equivalent to $\forall k. \, \exists m. \, f_m \geq {\omega \cdot k}$, and since $\exists m. \, f_m \geq \omega \cdot k$ is semidecidable by \cref{thm:CCimpliesOmega.kTobeSemidec}, this expression is of the claimed form.
\end{proof}

A consequence is that under countable choice, $\omega^2$-decidable propositions are closed under countable meets:
\begin{theorem}[\flink{Theorem-46}]
	Let $P : \N \to \Prop$ be a family of $\omega^2$-decidable propositions.
	Assuming countable choice, $\forall n. \, P_n$ is $\omega^2$-decidable.
\end{theorem}
\begin{proof}
	By the previous result, for each $n$, there exists a family $Q^n : \N \to \Prop$ of semidecidable propositions such that $P_n \leftrightarrow \forall i. \, Q^n_i$.
	Using countable choice, we can formulate this as the existence of a family $Q: \N \times \N \to \Prop$.
	Composed with any bijection $\N \times \N \cong \N$, this gives a family $\overline Q : \N \to \Prop$ of semidecidable propositions with the property $\forall n. \, P_n \leftrightarrow \forall n. \, \overline Q_n$. The right-hand side is $\omega^2$-decidable by \cref{thm:forallPnomega2}.
\end{proof}

\subsection{Avoiding Countable Choice with Sierpi\'nski-Semidecidability}\label{subsec:Sierpinski}

If one wants to study a notion of semidecidability that is closed under countable joins without the assumption of countable choice, one can consider the following:

\begin{definition}[\flink{Definition-47} Sierpi\'nski type]\label{def:sierpinski}
	The \emph{Sierpi\'nski type}, denoted by~$\Sierp$, is the free $\omega$-complete partial order on the unit type \cite{veltri2017type,chapman2019quotienting}.
	It can be implemented as a higher inductive type with the constructors $\top$, $\bot$, and $\Svee : (\N \to \Sierp) \to \Sierp$, together with the expected equations.
\end{definition}

	Alternatively, $\Sierp$ can be defined as the initial $\sigma$-frame.
	It is a sub-$\sigma$-frame of $\Prop$, and the canonical function $\Sierp \to \Prop$ is given by $s \mapsto s = \top$; see Escard\'o~\cite{EscardoQD} for more details.
The induced notion of semidecidability has previously been used by Gilbert~\cite{gilbert2017formalising} and goes back to~\cite[p.~53]{synthetictopology}.
\begin{definition}[\flink{Definition-48} Sierpi\'nski-semidecidability]
	A proposition $P$ is said to be \emph{Sierpi\'nski-semidecidable} if
	\begin{equation*}
		\Sigma s : \Sierp. \, P \leftrightarrow s = \top.
	\end{equation*}
\end{definition}

Note that even though formulated with $\Sigma$ rather than $\exists$, the type expressing Sierpi\'nski-semidecidability is still a proposition, since for $s, t : \Sierp$, we have $s = t$ if and only if $s = \top \leftrightarrow t = \top$.
The following is standard:

\begin{lemma}[\flink{Lemma-49}]
	Assuming countable choice, semidecidability and Sierpi\'nski-semidecidability coincide. \qed
\end{lemma}

One direction holds without assuming countable choice:
\begin{lemma}[\flink{Lemma-50}]\label{lem:semidec-implies-sierpinski}
	If a proposition $P$ is semidecidable, then it is also Sierpi\'nski-semidecidable.
\end{lemma}
\begin{proof}
	Given a binary sequence, we replace all occurrences of $0$ by $\bot$, and $1$ by $\top$.
	The countable join of this sequence is the required element of $\Sierp$.
\end{proof}

However, the following
shows that Sierpi\'nski-semidecidability is weaker than semidecidability, as the analogous statement for semidecidability implies 
Escard\'o-Knapp choice:

\begin{proposition}[\flink{Proposition-51}]\label{prop:sierpinski-decidability-of-comparison}
	If a proposition is $(\omega \cdot n + k)$-decidable, then it is Sierpi\'nski-semidecidable.
\end{proposition}
\begin{proof}
	It suffices to show that the proposition $\alpha \geq \omega \cdot n + k$ is Sierpi\'nski-semidecidable.
	By comparing the finite part of $\alpha$ with $k$, this inequality can be seen to either be equivalent to $\alpha \geq \omega \cdot n$ or to $\alpha \geq \omega \cdot (n+1)$, which means that we can freely assume $k=0$.
	Therefore, it suffices to construct a function $s : \N \to \Brw \to \Sierp$, where we write $s \, n \, \alpha$ as $s_n(\alpha)$, with the property
	\begin{equation}\label{eq:spec-for-sierpinski-s}
		s_n(\alpha) = \top \, \longleftrightarrow \, \alpha \geq \omega \cdot n.
	\end{equation}
	As the proposition $\alpha \geq \omega \cdot n$ is decidable for $n=0$ and $n=1$,
	we can simply define $s_n(\alpha)$ to be $\bot$ or $\top$ accordingly in these cases.
	Otherwise, we do induction on $\alpha$:
        \begin{gather*}
		s_{n+2}(\bzero) \defeq \bot,\quad
		s_{n+2}(\bsuc \beta) \defeq s_{n+2}(\beta),\quad
		s_{n+2}(\blim f) \defeq \Svee_{i:\N} s_{n+1}(f_i).
        \end{gather*}
	It is easy to see that bisimilar sequences are mapped to equal values.
	The property~\eqref{eq:spec-for-sierpinski-s} follows from the construction of $\Sierp$ as the initial $\sigma$-frame; for example, we have
	\begin{align*}
		s_{n+2}(\blim f) = \top \; &\leftrightarrow \; \left(\Svee_{i:\N} s_{n+1}(f_i)\right) = \top \\
		&\leftrightarrow \; \exists i. \, s_{n+1}(f_i) = \top \\
		&\leftrightarrow \; \exists i. \, f_i \geq \omega \cdot (n+1) \\
		&\leftrightarrow \; \blim f \geq \omega \cdot (n+2). \qedhere
	\end{align*}
\end{proof}

\section{On the Agda Formalization}\label{sec:Agda-technical-discussion}

Overall, the formalization in Agda closely follows the paper development,
although we consider two techniques to be worth commenting on.

\subsection{A Relational Approach to Maps out of Higher Inductive Types}
As discussed in \cref{subsec:limmin-constructive}, the construction of the
\(\limMin\) function is subtle.
The point is that $\Brw$ is a quotient inductive-inductive type, so
that, besides defining \(\limMin\) on point constructors, we also need
to define its action on the two path constructors for bisimilarity and
set truncation.

In order to get a well-defined map, it moreover becomes necessary to
simultaneously prove that \(\limMin\) is monotone, so that we must carry out a
mutually recursive construction.
The monotonicity proof of \(\limMin\) also requires dealing with path
constructors, and this becomes combinatorially involved, since \(\limMin\) takes
two arguments, so that monotonicity in both arguments leads us to consider four
arguments of type \(\Brw\).
If two of these arguments are limits and two of them arise from the \(\bbisim\)
path constructor, then we must fill a square (as there are two interval
variables arising from the two path constructors). Similarly, if three, resp.\
four of these arguments arise from the \(\bbisim\) constructor, then we must fill
a three-, resp.\ four-dimensional cube. Moreover, there are over three dozen
cases that involve path constructors, making this approach intractable.

It is natural to ask whether we cannot dispel these path constructor cases by
appealing to the fact that type expressing monotonicity of \(\limMin\) is a
proposition, so that any cubes can always be filled.
The problem with this idea is that we could not get it to pass Agda's
termination checker.

The alternative approach that we ended up adopting is to first define a
\emph{relation}, prove it to be single-valued and total, and to extract a
function only at the very end.
The idea that inspired this approach is that all our trouble is caused by path
constructors, so that we simply define a relation on point constructors only,
which, as a type family, automatically respects equality (and hence path
constructors).
Relational approaches to recursive functions have been considered before, see
e.g.~Bove and Capretta~\cite{BoveCapretta2005}, and what we describe here may be taken as an
illustrative example of how it is beneficial for recursive functions
on higher inductive types too.

Specifically, we define, for \(x,y,z : \Brw\), a type \(\limMinRel x y z\) and
simultaneously prove that this relation is monotone in both arguments, i.e.,
if \(x \leq x'\) and \(y \leq y'\), then
\begin{fequation}{Equation-9}
  \limMinRel x y z \text{ and } \limMinRel{x'}{y'}{z'} \text{ implies } z \leq z'.
\end{fequation}
The intended reading of \(\limMinRel x y z\) is that it is inhabited precisely
when \(z = \limMiF x y\).

Now, \(\limMinRel{x}{y}{z}\) is defined as an inductive type family, where
would-be recursive function calls are recorded as hypotheses. For example, since
\eqref{eq:item-blim-bsuc} specifies that
\[
  \limMiF {(\blim\,f)}{(\bsuc\,y)} = \limMiF {(\blim\,f)} {y},
\]
we add a
constructor
\begin{align*}
  \limsucc &: (f : \N \toincr{<} \Brw) \, (y,z : \Brw) \\
  &\hspace*{-1mm}\to \limMinRel{\blim f}{y}{z} \to \limMinRel{\blim f}{\bsuc y}{z}.
\end{align*}

In total, \(\limMinRel x y z\) has only six constructors. We consider the cases
where the first argument is zero, a successor, or a limit. Only in the latter
case do we further case split on the second argument. Finally, we add a
truncation constructor that makes \(\limMinRel x y z\) proposition-valued.

As mentioned, we simultaneously prove that this relation is monotone and we note
that, in contrast to the direct functional approach discussed earlier, this does
\emph{not} require us to deal with any complicated path constructors.
This completes the mutually inductive part of the construction.

A convenient and direct consequence of monotonicity is that the relation is
single-valued: if \(\limMinRel x y z\) and \(\limMinRel x y z'\) both hold, then
\(z = z'\).
Hence, for any two \(x,y : \Brw\), the type \(\Sigma z : \Brw. \,\limMinRel x y z\)
is a proposition.

Totality of the relation expresses that we have an element of this type for any \(x,y : \Brw\).
Thanks to the fact that the type is a proposition, we can prove this inductively
by \emph{only} considering point constructors for \(x\) and \(y\), and this is
easy thanks to the constructors of \(\limMinRel x y z\).

In summary, the type \(\Sigma z : \Brw. \, \limMinRel x y z\) is contractible for
every \(x,y : \Brw\), and the first projection of the type yields a function
\(\limMin : \Brw \to \Brw \to \Brw\) that moreover has the right computational
behaviour, e.g.
\[
  \limMiF {(\blim\,f)}{(\bsuc\,y)} = \limMiF {(\blim\,f)} {y}
\]
holds definitionally, as well as all other expected equations. This is obviously convenient in the formalization.

\subsection{Properties of Maps out of Truncated Types}
Recall that in \cref{subsec:characteristic-ordinal} we define a function out of a
truncated type via a weak constancy argument.
Abstracting away from specifics, consider such a function $g : \trunc{A} \to B$ defined via the weak constancy of $f : A \to B$. The situation is as depicted in the following
commutative diagram:
\[
  \begin{tikzcd}
    A \ar[rr,"f"] \ar[dr,"\totrunc{-}"'] & & B \\
    & \trunc{A} \ar[ur,"g"']
  \end{tikzcd}
\]
When proving a property of values of \(g\), we may assume that we have an actual
element \(a\) of \(A\) and prove the relevant property for \(f\,a\) instead.
Rather than repeating this setup for various properties, we state, prove and
appeal to a ``meta''-lemma: for any proposition-valued family \(P\) over \(B\),
we have \(P(g\,t)\) for all \(t : \trunc{A}\), as soon as \(P(f\,a)\) holds for
all \(a : A\).
One can similarly define and prove a binary version for families indexed over
\(B \times B\), which in the setting of \cref{subsec:characteristic-ordinal} allows
one to easily prove a monotonicity property for the map out of the truncated type.

In the formalization, this meta-lemma for the characteristic ordinal
construction is called
{\small\texttt{characteristic-ordinal-up-to-reduction-lemma}} \flink{meta-lemma} (with its
binary version having a subscript {\small\texttt{2}}).

\section{Conclusions}\label{sec:conclusions}

We have introduced and studied the notion of $\alpha$-decidability, which generalizes existing notions.
We view this paper as the first step towards a theory of generalized decidability, and at the moment, many questions remain open.
In particular, we have not answered the question how the various notions of decidability assemble in a hierarchy, in the sense of whether $\alpha$-decidable propositions are always $\beta$-decidable as long as $\alpha \leq \beta$.
We have only given first results in that direction (cf.~\cref{subsec:first-steps-towards-hierarchy}).

A further question is whether it is possible to calculate implications between propositions:
If $P$ is $\alpha$-decidable and $Q$ is $\beta$-decidable, what can be said about the proposition $P \to Q$?
The important special case that arises when $Q$ is empty is partially answered by \cref{cor:neg-of-semidec-is-omega2}, which says that the negation of a semidecidable proposition is $\omega^2$-decidable.

In \cref{sec:semidec-and-quantifiers}, we work with a family $P: \N \to \Prop$ of semidecidable propositions.
While this might very well be the most interesting case, more general statements would involve families of $\alpha$-decidable propositions.
We do not know what can be said in the general case, but if one could, for example, show that taking countable join raises the ``decidability level'' by $\omega$, then this would give a more direct argument for \cref{thm:existsforallSemidec}.

A deeper but justified question is whether the type $\Brw$ of Brouwer tree ordinals is the best possible notion of ordinals for our purposes,
as any notion of ordinals gives rise to a theory of $\alpha$-decidability.
As we have discussed in the introduction, some of these theories collapse the hierarchy of decidability (and are thus less interesting). This happens, for example, for ``purely syntactic'' notions that represent Cantor Normal Forms and ensure that $\leq$ is always decidable, and for ``purely semantic'' notions that cause $\leq$ to not have any decidability properties.
Between these two extremes, there might be notions alternative to $\Brw$ that give rise to interesting theories of $\alpha$-decidability, with properties different from ours.

As discussed in \cref{subsec:Sierpinski}, switching from binary sequences (semidecidability) to the initial $\sigma$-frame (Sierpi\'nski-semidecidability) makes it possible to avoid the assumption of countable choice in certain situations.
We have not managed to formulate a variation of $\Brw$ with the same property, but it is possible that such a variation exists.



\bibliography{references}

\end{document}